\documentclass{lmcs}

\usepackage{graphicx} 
\usepackage{color,tikz}
\usepackage{amssymb}
\usepackage{mathtools}

\newcommand{\adom}{\textup{adom}}

\usepackage{xspace}
\newcommand{\dom}{\textup{cod}\xspace}
\newcommand{\sdidist}{\textup{sdi-dist}\xspace}
\newcommand{\sdqdist}{\textup{sdq-dist}\xspace}
\newcommand{\editdist}{\textup{edit-dist}\xspace}
\newcommand{\core}{\textup{core}\xspace}

\newcommand{\problem}[1]{\textup{\sc #1}}

\newcommand{\calS}{\mathcal{S}}
\newcommand{\extension}[1]{[\![#1]\!]}

\title{Query Repairs}
\titlecomment{An extended abstract of this paper appeared in the proceedings of ICDT 2025. The present paper contains detailed proofs that were absent in the extended abstract; fixes several mistakes and imprecisions; and closes a  gap by providing a matching upper bound for $\preceq^{\editdist}$-\problem{repair-verification} (cf.~Theorem~\ref{thm:editdistrepaircomplexity}(4)).}

\thanks{The first author was supported by EU Horizon  2020 Grant MSCA-101031081. The second author was partially supported by NSF Grant IIS-1814152. The third author was supported by the DFG Collaborative Research Center 1320 EASE}	

\author[B.~ten Cate]{Balder ten Cate\lmcsorcid{0000-0002-2538-5846}}[a]
\author[Ph.~Kolaitis]{Phokion Kolaitis\lmcsorcid{0000-0002-8407-8563}}[b]
\author[C.~Lutz]{Carsten Lutz\lmcsorcid{0000-0002-8791-6702}}[c]

\address{University of Amsterdam, The Netherlands}	
\email{b.d.tencate@uva.nl}  

\address{UC Santa Cruz \& IBM Research - Almaden, USA}	
\email{kolaitis@soe.ucsc.edu}  

\address{Leipzig University and Center for Scalable Data Analytics and Artificial Intelligence (ScaDS.AI), Dresden/Leipzig, Germany}
\email{clu@informatik.uni-leipzig.de}


\begin{document}

\begin{abstract}
    We formalize and study the problem of repairing database queries based on user feedback in the form of a collection of labeled examples. We propose a framework based on the notion of a proximity pre-order, and  investigate and compare query repairs for conjunctive queries (CQs) using different such pre-orders. The proximity pre-orders we consider are based on query containment and on 
    distance metrics for CQs.
\end{abstract}

\maketitle

\section{Introduction}
    
When querying a database, it may happen that the query result includes some undesired tuples and/or that some desired  tuples are missing. In such cases, it is often necessary to adjust
the query to ensure that the result aligns with expectations, i.e., it includes the  desired tuples and omits the undesired ones. 

\begin{figure}[t]
    \centering
    \begin{tabular}{|lll|}
    \hline
    \multicolumn{3}{|c|}{Release} \\
    \hline
    Babygirl & 2025 & DE \\
    Babygirl & 2025 & FR \\
    Nosferatu & 2025 & DE \\
    Nosferatu & 2024 & FR \\
    \ldots && \\
    \hline
    \end{tabular}
    \caption{Example instance}
    \label{fig:example-db}
\end{figure}

\begin{exa} Consider the database instance $I$ in Figure~\ref{fig:example-db}.
    A user,  wanting to retrieve movies released in both Germany and France, issues the 
    query
    \[q(x) \colondash \text{Release}(x,y,\text{FR}), \text{Release}(x,y,\text{DE}).\] The query results
    include \text{Babygirl} but not \text{Nosferatu}.
    The user spots the latter as a missing answer, and wants to revise the query. A solution is to change the 
    query to     \[q'(x) \colondash \text{Release}(x,y,\text{FR}), \text{Release}(x,z,\text{DE}).\] A more radically different query such as \[q''(x) \colondash \text{Release}(x,y,\text{FR})\] would also account for the missing answer, but clearly fails to capture the user's intention. 
\end{exa}

    
We propose a formalization of the above problem through the notion of \emph{query repairs}, as follows. We assume that we are given a query $q$ and a set of \emph{labeled examples},
by which  we mean pairs $(I,\textbf{a})$ with $I$ a 
database instance  and $\textbf{a}$ a tuple of values
from the active domain of $I$, labeled as positive or negative to indicate whether $\textbf{a}$ is 
desired or undesired as an answer on input $I$. 
In the above example, for instance, the input query is $q$ and there is a positively labeled example $(I,\text{Nosferatu})$ that $q$ fails to fit.
A \emph{query repair}, then, is a query $q'$ that fits the given
 labeled examples  and ``differs from $q$ in a minimal way''.
Different notions of query repair arise by using different means to formalize what it means for two queries to differ in a minimal way.
Besides requiring that $q'$ fits the labeled examples and 
differs minimally from $q$, depending on the context,
it may  be natural to additionally require that $q\subseteq q'$ or that $q'\subseteq q$. This leads to further refinements of the notion of a query repair, namely \emph{query generalization} and \emph{query specialization}, respectively, which we also investigate.



We propose a broad framework for defining what it means for
two queries to differ in a minimal way, based on a 
\emph{proximity pre-order} $\preceq$, i.e.,  a family of
pre-orders $\preceq_q$ (one for each query $q$), where  
$q' \preceq_q q''$ asserts that query $q'$ is at
least as close to $q$ as $q''$. A query $q'$ is then a \emph{$\preceq$-repair} of a query $q$
if $q'$ fits the given labeled examples and there is no query $q''$ with 
the same property  such that $q'' \prec_q q'$. We instantiate this
framework for conjunctive queries (CQs), focusing
mainly on two kinds of 
%
proximity pre-orders:
the \emph{containment-of-difference} proximity pre-order $\preceq^{\dom}$,  based on query containment, and the \emph{edit-distance} proximity pre-order $\preceq^{\editdist}$,  based on a  distance metric between queries defined in terms of a suitably adapted version of edit distance. To be more
precise,  $q_1 \preceq^{\dom}_q q_2$ if for every instance $I$, the symmetric difference of $q_1(I)$ and $q(I)$ is contained in the symmetric difference of $q_2(I)$ and~$q(I)$. Moreover, $q_1 \preceq^{\editdist}_q q_2$ if the edit distance between
the homomorphism core of $q_1$ and the homomorphism core of $q$ is no larger
than the edit distance between
the homomorphism core of $q_2$ and the homomorphism core of $q$, modulo
variable renaming. Here, edit distance means the minimum number of additions and deletions of atoms required to make the  queries identical.


\begin{exa}[Generalization]
\label{ex:running-generalization} 
Consider the CQ
\[q(x) \colondash R(x,y), R(y,z), R(z,u), R(u,x),\] which returns all values that lie on a directed $R$-cycle of length 4. Let $I$ be the instance that consists of the 
facts $R(a,b), R(b,c), R(c,a)$, i.e., $I$ is the directed $R$-cycle of length 3. Clearly, 
$a\not\in q(I)$. Let $E$ be the singleton set of examples consisting  of  $(I,a)$ labeled as a positive example.
Which CQs qualify as repairs for $(q,E)$ or as generalizations for $(q,E)$? Note that since we are given only positive examples, specializations do not seem to be a natural choice here.

It will turn out that there are multiple $\preceq^{\dom}$-repairs for
$(q,E)$. In particular, these include
\begin{itemize}
    \item the CQ expressing that $x$ lies on a directed $R$-cycle of length 12, and
    \item the CQ expressing that $x$ lies on a directed 
$R$-cycle of length 3. 
\end{itemize}
Both of these are reasonable options. If
we ask for  $\preceq^{\dom}$-generalizations for $(q,E)$, then, of these two, only
the first remains.

Considering edit-distance-based repairs, there are precisely four $\preceq^{\editdist}$-
repairs of $(q,E)$, each obtained from $q$ by dropping a different atom from the body. Also these are reasonable options. The
same four CQs are also the $\preceq^{\editdist}$-generalizations for
$(q,E)$.
\end{exa}

\begin{exa}[Specialization]
\label{ex:running-specialization}
    Consider the CQ
    \[q(x) \colondash R(x,y), R(y,z),\] 
    returning all values that have an outgoing $R$-path of length $2$.
    Let $E$ be the set consisting of 
    \begin{itemize}
        \item a 
    negative example $(I,a)$ with $I=\{R(a,b), R(b,c)\}$, and
    \item
    a positive example $(J,a)$ with $J=\{R(a,b), R(b,c), R(c,d)\}$.
    \end{itemize}
  Which CQs qualify as specializations for $(q,E)$? In the same
  way in which generalizations are linked closely to the positive
  examples, specializations are linked closely to the negative
  examples. Note, however, that  by itself the negative example in $E$ does not provide much guidance as to what would be a ``good repair'' as there are many possible options. The positive example gives 
    (in this case, quite specific) additional guidance regarding the ``direction'' towards which we should look to find the repair.

    It will turn out that there is precisely
    one $\preceq^{\editdist}$-specialization, namely the very natural CQ \[q'(x) \colondash R(x,y), R(y,z), R(z,u).\] 
   However,  $q'$ does not qualify as a $\preceq^{\dom}$-specialization, since the CQ
    \[q''(x) \colondash R(x,y), R(y,z), R(u,z), R(u,v), R(v,w)\] also fits
    and is ``closer'' to $q$ in terms of query containment.
    In fact, as we will see,  no $\preceq^{\dom}$-specialization for $(q,E)$ exists. 
\end{exa}



The problem of constructing a query that fits a given 
set of labeled data examples has been studied extensively and is known under different names such
as \emph{reverse engineering}, \emph{query learning}, or \emph{fitting}; see for instance~\cite{SigmodRecordColumn} which
offers a comparison of several  fitting algorithms for CQs. 
Recently, in~\cite{pods2023:extremal}, \emph{extremal} variants of the fitting problem for CQs were studied, including 
 (weakly/strongly) \emph{most-general fitting} 
 and  \emph{most-specific fitting}. There, the input consists of a set of positive and negative examples and the task is to find a most specific CQ, or a most general CQ, that fits them. 
 We can think of such extremal fitting problems as 
 constrained
 versions of the fitting problem for CQs where an additional requirement is put on the output query. In the same spirit, the \emph{query repair} problem
 can also be viewed as a constrained version of fitting where the input now includes, in addition, a CQ $q$,
 and the output is required to be a fitting CQ that differs minimally from $q$.\footnote{For the trivial proximity pre-order $\preceq_q$ relating every CQ to every  CQ, the query repair problem coincides with the fitting problem.} As a part of our contributions, we will establish close relationships between query repair problems and extremal fitting problems.

\paragraph{Overview of contributions}
In Section~\ref{sec:repairs}, we formally define $\preceq$-query repairs, as well as $\preceq$-generalizations and $\preceq$-specializations, based on a given proximity pre-order $\preceq$. We also propose, for each of these, three algorithmic problems: \emph{verification}, \emph{existence} and \emph{construction}.
The remaining sections focus specifically on CQs.

In Section~\ref{sec:containment-based}, we study the containment-based proximity pre-order $\preceq^{\dom}$. Besides examples of the resulting notions of generalization, specialization, and repair, our results, here, include: 
\begin{enumerate}
    \item[(a)] structural characterizations that relate
$\preceq^{\dom}$-generalizations to most-specific fittings
(Theorem~\ref{thm:dom-generalization-most-specific}) and
$\preceq^{\dom}$-specializations to weakly most-general fittings
(Theorems~\ref{thm:reduction-specialization} and~\ref{thm:wmg-as-dom-specialization}), that latter being restricted to Boolean queries in one direction. The characterizations imply that there is always a unique $\preceq^{\dom}$-generalization (unless no suitable fitting CQ exists) while $\preceq^{\dom}$-specializations do not always exist.
    \item[(b)] based on this, results that identify the computational complexity of the verification, existence, and construction of $\preceq^{\dom}$-generalizations and
    Boolean $\preceq^{\dom}$-specializations. 
    \item[(c)] results that relate $\preceq^{\dom}$-repairs to 
    $\preceq^{\dom}$-specializations and $\preceq^{\dom}$-generalizations,  allowing us to apply some of the above algorithmic results to the more general case of 
    $\preceq^{\dom}$-repairs. However, we also illustrate that the 
    behavior of $\preceq^{\dom}$-repairs is often counterintuitive.
    For instance, $\preceq^{\dom}$-repairs need not exist and also there can be infinitely
    many $\preceq^{\dom}$-repairs.
    In contrast to $\preceq^{\dom}$-generalizations and
    $\preceq^{\dom}$-specializations, $\preceq^{\dom}$-repairs thus do
    not seem to be very natural.
\end{enumerate}

In Section~\ref{sec:distance-based}, we study 
proximity pre-orders based on distance metrics. In particular, we 
propose a
 distance metric for CQs based on edit distance that gives
rise to a proximity pre-order $\preceq^{\editdist}$. We show
that there is always a non-empty and finite set of $\preceq^{\editdist}$-repairs (respectively,  $\preceq^{\editdist}$-generalizations, and $\preceq^{\editdist}$-specializations),
unless the given examples do not admit any fitting candidate query of the corresponding kind.
Moreover, we 
shed light on the complexity of the existence and verification problems (Theorem~\ref{thm:editdistrepaircomplexity}).
 We also
show that other seemingly natural distance metrics lead to repair
notions that behave worse.


\paragraph{Outline}
Section~\ref{sec:prelim} contains technical preliminaries. In Section~\ref{sec:repairs},
we define query repairs.
In Section~\ref{sec:containment-based}, we explore containment-based query repairs. 
In Section~\ref{sec:distance-based}, we explore  query repairs based on distance metrics.
We conclude in Section~\ref{sec:discussion} with a discussion of future directions.

\paragraph{Related work}
Our notion of \emph{query repairs} is in part inspired by the literature on \emph{database repairs} introduced in \cite{DBLP:conf/pods/ArenasBC99}.
There, one is given a database $D$ that is inconsistent in the sense that it 
violates one or more integrity constraints and the aim is to answer a given
query over all possible repairs of $D$, that is, all  databases consistent
with the integrity constraints that ``differ from $D$ in a minimal way''. Different notions of repairs, including set-based repairs and cardinality repairs,  arise by 
formalizing in different ways what it means for two databases to ``differ in a minimal way''. 
Research in this area has been rather active and fruitful~\cite{DBLP:series/synthesis/2011Bertossi}.

There is extensive literature on approximating a query $q$ by some other query $q'$ such that $q$ is contained in $q'$ or $q'$ is contained in $q$. In the former case $q'$ is often called an \emph{upper approximation} or an \emph{upper envelope} of $q$, while in the latter case it is called a \emph{lower approximation}, a \emph{lower envelope}, or a \emph{relaxation} of $q$. 
For instance, \cite{DBLP:conf/aaai/MusleaL05} proposes an algorithm for relaxing the 
where clause of an over-constrained database query that returns an empty result. 
Naturally, one is interested in optimal (with respect to containment) such approximations, which are known as \emph{tight} upper or lower envelopes. Lipski \cite{DBLP:journals/tods/Lipski79} studied upper and lower approximations in the context of databases with incomplete information, while Libkin \cite{DBLP:journals/tcs/Libkin98} carried out a study of formal models of approximation in databases. A related body of work focused on the problem of using approximation to achieve more efficient query evaluation. In particular, approximations of Datalog queries by CQs or unions of CQs were investigated in \cite{DBLP:conf/pods/Chaudhuri93,DBLP:journals/jcss/ChaudhuriK97}. More recently, approximations of CQs by CQs  of tractable combined complexity (such as acyclic CQs or CQs of bounded treewidth) were studied in \cite{DBLP:journals/siamcomp/BarceloL014,Barcelo2020:static}. In a different, yet related direction, tight  lower envelopes were used in the area of answering queries using views \cite{DBLP:conf/pods/DuschkaG97,DBLP:conf/cikm/KantereOKS15}, where such envelopes approximate a perfect rewriting.
Upper and lower envelopes were also 
used as tractable approximations of the answers to ontology-mediated queries, both over consistent databases \cite{DBLP:conf/kr/HagaLSW21} and over 
inconsistent ones \cite{DBLP:conf/ijcai/BienvenuR13}.

The literature on approximations summarized above is based on the notion of containment of one query to another. Notions of  ``closeness'' or ``similarity'' of queries that are not based on containment   have also been investigated. For example, a notion of closeness based on  suitable combinations of precision and recall was used to study the problem of translating a query over some schema to a semantically similar query over a different schema \cite{DBLP:journals/vldb/ChangG01}.  Furthermore, 
a notion of semantic similarity of queries based on available query logs was explored in \cite{DBLP:conf/sigir/BordinoCDG10}.


In the area of \emph{belief revision}, 
a number of proposals have been made for \emph{model-based} revision and update operators, in which a knowledge base is viewed semantically as a set of \emph{possible worlds} (where a world is a propositional truth assignment), 
and update/revision is performed on sets of possible worlds. Various concrete update and revision operators have been proposed
based on different notions of relative proximity for possible worlds, including using Hamming distance \cite{Dalal1988:investigations,Forbus1989:introducing} and containment-of-difference \cite{Satoh1988:nonmonotonic,Winslett1990:updating}.

In software engineering,
\emph{automated program
repair} techniques seek to aid
developers by suggesting likely correct
patches for software bugs. They take as input a
program and a specification of 
correctness criteria that the fixed program should meet. Most techniques assume 
that the correctness criteria are given by means of a test suite:
one or more failing tests indicate a bug
to be fixed, while passing tests indicate
behavior that should not change. The desired output
is a set of program changes
that leads all
tests to pass. See~\cite{LeGoues2019:automated} for an overview.


\section{Preliminaries}
\label{sec:prelim}

\paragraph{Database Schemas, Instances, Queries, and Data Examples}

As usual, a schema $\mathcal{S}$ is a set of relation symbols, each with associated arity. A \emph{database instance} over $\mathcal{S}$ is a finite set $I$  of \emph{facts} of the form $R(a_1,\dots,a_n)$ where $R \in \mathcal{S}$ is a relation symbol of arity $n$ and $a_1,\dots,a_n$ are \emph{values}. We use $\mathit{adom}(I)$ to denote the set of all values used in~$I$.
We can  then view a \emph{query} over a schema~$\mathcal{S}$, semantically, as
a function $q$ that maps each 
database instance $I$ over $\mathcal{S}$ to a 
set of $k$-tuples $q(I)\subseteq \mathit{adom}(I)^k$, where $k\geq 0$ is the \emph{arity} of the query.
A query of arity zero is called a \emph{Boolean} query. 
We write $q_1\subseteq q_2$ and say that $q_1$ is \emph{contained} in $q_2$ if 
$q_1(I)\subseteq q_2(I)$ for all database instances $I$.
Two queries $q_1$ and $q_2$ are \emph{equivalent}, written $q_1 \equiv q_2$, if $q_1\subseteq q_2$ and $q_2\subseteq q_1$.

A \emph{data example} for a $k$-ary query $q$ consists
of a database instance $I$ together with a
$k$-tuple of values. 
We denote by $\extension{q}$ the set of all
data examples $(I,\textbf{a})$ for which it 
holds that $\textbf{a}\in q(I)$.
A \emph{labeled example}
is a data example that is labeled as positive or as negative. By a \emph{collection of labeled examples} we mean a pair 
$E=(E^+,E^-)$, where $E^+$ and
$E^-$ are sets of examples. Here, the data examples in  $E^+$ are considered as
positive examples, and the data examples in $E^-$ are considered as negative examples.
A query $q$ \emph{fits} $E=(E^+,E^-)$ if $\textbf{a}\in q(I)$ for each $(I,\textbf{a})\in E^+$
 and
$\textbf{a}\not\in q(I)$  for each 
$(I,\textbf{a})\in E^-$. In other words, 
$q$ fits $E=(E^+,E^-)$ if 
$E^+\subseteq\extension{q}$ and 
$E^-\cap\extension{q}=\emptyset$.
Here, we assume that $q$ has the same
arity as the data examples in $E^+$ and $E^-$.
We will often abuse notation and
write that $q$ fits $E^+$ (or that $q$ fits $E^-$), 
meaning that $q$ fits $(E^+,\emptyset)$ (respectively, 
$q$ fits $(\emptyset,E^-)$).


\paragraph{Conjunctive Queries}

We will be focusing specifically on 
conjunctive queries. By a $k$-ary \emph{conjunctive query (CQ)} over a schema $\mathcal{S}$, we  mean an expression of the form 
$q(\textbf{x}) \colondash \alpha_1,\ldots,\alpha_n$ 
where $\textbf{x}=x_1, \ldots, x_k$ is a sequence of variables and each $\alpha_i$ is a relational atom that uses a relation symbol from $\mathcal{S}$ and no constants.
Note: the restriction to queries without constants is not
essential for our results (cf.~\cite[Remark 2.3]{SigmodRecordColumn}) but simplifies the presentation.

The variables in $\textbf{x}$ are called \emph{answer variables} and
the other variables used in the atoms~$\alpha_i$ are the \emph{existential variables}.
Each answer variable is required to occur in at least one
atom $\alpha_i$, a  requirement  known as the \emph{safety condition}.

For CQs $q,q'$ of the same arity, we denote their \emph{conjunction} by  $q\land q'$. For instance,
the conjunction of
\[q_1(x,y)\colondash R(x,y,z)\] and \[q_2(x,x)\colondash S(x,z)\]
is 
\[q(x,x)\colondash R(x,x,z), S(x,z').\]
Formally,
let $q_1(x_1, \ldots, x_k)$ and $q_2(y_1, \ldots, y_k)$ be two CQs of the same arity that (without loss of generality) do not have any variables in common.  Then
$q_1\land q_2$ denotes the CQ constructed 
as follows. 
First, let $\sim$ denote the smallest
equivalence relation on $\{1, \ldots, k\}$
such that $i\sim j$ holds whenever $x_i=x_j$
or $y_i=y_j$. Next, let $z_1, \ldots, z_k$ be
a sequence of fresh variables such that $z_i=z_j$ iff
$i\sim j$. Finally, $q_1\land q_2$ denotes the
query 
$q'(z_1, \ldots, z_k) \colondash \phi$,
where $\phi$ is obtained by concatenating the 
bodies of $q_1$ and $q_2$ and replacing all
occurrences of answer variables $x_i$ respectively $y_i$ by $z_i$ (for all $i\leq k$).
The following lemma, which is easy to prove, justifies our use of the notation
$q_1\land q_2$.

\begin{lem}
   For all CQs $q_1, q_2$, it holds that 
    $\extension{q_1\land q_2}=\extension{q_1}\cap \extension{q_2}$
\end{lem}

With the \emph{size} of a CQ, denoted $|q|$, we mean
the number of atoms in it. 
The query output $q(I)$ is defined as usual, cf.~any standard database textbook.

\paragraph{Canonical example, Homomorphism, Core
}
Every CQ $q(x_1, \ldots, x_k)$ has a \emph{canonical example}
$e_q$, namely the data example $(I_q,\langle x_1, \ldots, x_k\rangle)$, where $I_q$ is the database instance (over the same schema 
as $q$) whose active domain
consists of the variables  in $q$ and whose 
facts are the atomic formulas
in $q$.

Given data examples $e=(I,\textbf{a})$ and $e'=(J,\textbf{b})$ over the same schema and with the same number of distinguished elements, a 
\emph{homomorphism} $h: e\to e'$ is a map from $\adom(I)$ to $\adom(J)$ that  preserves all facts and such that $h(\textbf{a})=\textbf{b}$. When such a homomorphism exists, we   say that $e$ ``homomorphically maps to'' $e'$ and write
$e\to e'$. 
We say that $e$ and $e'$ are \emph{homomorphically equivalent} if
$e\to e'$ and $e'\to e$. 
It then holds that $e\in \extension{q}$ iff $e_q\to e$.
Furthermore, 
the well-known Chandra-Merlin theorem states that 
$q\subseteq q'$ holds iff
$e_{q'}\to e_q$.

A data example $e$ is said to be a \emph{core} 
if every homomorphism $h:e\to e$ is surjective.
It is well-known that for every data example $e=(I,\textbf{a})$ there is a subinstance $I'\subseteq I$ such that $(I',\textbf{a})$ is a core and 
such that $(I,\textbf{a})$ and $(I',\textbf{a})$
are homomorphically equivalent. Moreover,
such  $(I',\textbf{a})$ is unique up to isomorphism, and may be referred to as 
\emph{the core of $e$}, denoted $\core(e)$. We  say that a 
CQ $q$ \emph{is a core} if its canonical example $e_q$ is a core.

\paragraph{Direct Products}

The \emph{direct product} of two database instances, denoted $I\times J$, is the database instance
containing all facts $R(\langle a_1,b_1\rangle, \ldots, \langle a_n,b_n\rangle)$ over the 
domain $adom(I)\times adom(J)$ such that $R(a_1, \ldots, a_n)$ is a fact of $I$ and $R(b_1, \ldots, b_n)$ is a fact of $J$. 
This naturally  extends to data examples: $(I,\langle a_1, \ldots, a_k\rangle)\times (J,\langle b_1, \ldots, b_k\rangle)=(I\times J, \langle \langle a_1, b_1\rangle, \ldots, \langle a_k,b_k\rangle\rangle)$. 

The following lemma expresses some 
well-known (and easily verified) properties of the direct
product construction:

\begin{lem}\label{lem:products}
    For all CQs $q$ and data examples $e_1,e_2$
    \begin{enumerate}
\item $e_1\times e_2\to e_1$ and $e_1\times e_2\to e_2$,
        \item 
$e_1\times e_2\in \extension{q}$ if and only if $e_1\in\extension{q}$ and $e_2\in\extension{q}$.
    \end{enumerate}
\end{lem}

 Lemma~\ref{lem:products} follows (using the Chandra-Merlin theorem) from well-known properties of direct products. We omit the proof.

\begin{lem}\label{lem:distributivity-boolean}
Let
    $q_1, q_2$ be Boolean CQs and let $I,J$ be instances. If $I\times I_{q_1}\to J$ 
    and $I\times I_{q_2}\to J$, then 
    $I\times I_{q_1\land q_2}\to J$.
\end{lem}

\begin{proof}
  For Boolean CQs $q_1, q_2$, it holds that $I_{q_1\land q_2}$ is isomorphic to the 
  disjoint union $I_{q_1}\uplus I_{q_2}$. Furthermore, it is easy
  to see that
  $I\times (I_{q_1}\uplus I_{q_2})$ is isomorphic to
  $(I\times I_{q_1})\uplus (I\times I_{q_2})$.
  It follows that  homomorphisms $h_1:I\times I_{q_1}\to J$ 
    and $h_2:I\times I_{q_2}\to J$ can be combined by unioning them, to obtain a homomorphism  
    $h:I\times I_{q_1\land q_2}\to J$.
\end{proof}

\section{Query Repairs}
\label{sec:repairs}

Fix a query language $\mathcal{L}$.
A \emph{proximity pre-order} $\preceq$ for $\mathcal{L}$ is a family of
pre-orders $\preceq_q$, one for every  $q\in\mathcal{L}$,  
satisfying the following conditions:
\begin{description}
    \item[\em Conservativeness] For all $q,q'\in\mathcal{L}$, $q\preceq_q q'$.
    \item[\em Syntax independence] Whenever $q'_1\equiv q_1$, $q'_2\equiv q_2$ and $q'_3\equiv q_3$, then 
    $q_1\preceq_{q_2} q_3$ iff $q'_1\preceq_{q'_2} q'_3$.
\end{description}

Let $q\in \mathcal{L}$, and let $E$ be a collection of labeled examples (of the same arity as $q$).
We call the pair $(q,E)$ an \emph{annotated $\mathcal{L}$-query}. The following
are the three main notions studied in this paper.
\begin{itemize}
    \item A \emph{$\preceq$-repair} for $(q,E)$
is a query $q'\in\mathcal{L}$ such that 
\begin{quote}\begin{itemize}
    \item[(i)] $q'$ fits $E$, and 
    \item[(ii)] there is no $q''\in\mathcal{L}$ with $q'' \prec_q q'$ that satisfies~(i).
\end{itemize}\end{quote}
\item A \emph{$\preceq$-generalization} for $(q,E)$
is a query $q'\in\mathcal{L}$ such that 
\begin{quote}\begin{itemize}
    \item[(i)] $q'$ fits $E$ and $q\subseteq q'$, and 
    \item[(ii)] there is no $q''\in\mathcal{L}$ with $q'' \prec_q q'$ that satisfies~(i).
\end{itemize}\end{quote}
\item A \emph{$\preceq$-specialization} for $(q,E)$
is a query $q'\in\mathcal{L}$ such that 
\begin{quote}\begin{itemize}
    \item[(i)] $q'$ fits $E$ and $q'\subseteq q$, and 
    \item[(ii)] there is no $q''\in\mathcal{L}$ with $q'' \prec_q q'$ that satisfies~(i).
\end{itemize}\end{quote}
\end{itemize}
(where $q_1\prec_q q_2$ is short for 
$q_1\preceq_q q_2$ and $q_2\not\preceq_q q_1$).

Note that under this definition, $\preceq$-specializations and $\preceq$-generalizations need not be $\preceq$-repairs.

Conservativeness and syntax-independence are minimal conditions on $\preceq$ needed 
to yield intuitive behavior for query repairs:
conservativeness ensures that if the input query $q$ already fits the given examples, then $q$ is its own $\preceq$-repair, while syntax independence ensures that 
equivalent queries have equivalent $\preceq$-repairs. 

In the rest of this paper, 
we will restrict attention to CQs. 
That is,
$\mathcal{L}$ is the class of CQs.

\begin{rem}\label{rem:no-other-relations}
We will restrict our attention to repairs, generalizations,
and specializations that use only the relation symbols which occur
in $(q,E)$. Note that when $E$ consists only of negative examples, a query repair or specialization could in principle contain relation symbols that do not occur in $(q,E)$. We 
will disregard such repairs. It is not difficult, however, to
adapt our results to the case where such additional symbols
would be admitted. 
\end{rem}

Several algorithmic problems  arise, all parameterized with a proximity pre-order $\preceq$.
 
\medskip

\par\noindent\fbox{\parbox{.98\columnwidth}{
\problem{$\preceq$-Repair-Verification} 

\smallskip

\textbf{Input}: An annotated CQ $(q,E)$ and a CQ $q'$ 

\textbf{Output}: \emph{Yes} if $q'$ is a $\preceq$-repair for $(q,E)$, \emph{No} otherwise
}}

\medskip

\par\noindent\fbox{\parbox{.98\columnwidth}{
\problem{$\preceq$-Repair-Existence} 

\smallskip

\textbf{Input}: an annotated CQ $(q,E)$

\textbf{Output}: \emph{Yes} if there is a $\preceq$-repair for $(q,E)$, \emph{No} otherwise
}}

\medskip

\par\noindent\fbox{\parbox{.98\columnwidth}{
\problem{$\preceq$-Repair-Construction} 

\smallskip

\textbf{Input}: an annotated CQ $(q,E)$ for which  
  a $\preceq$-repair exists
    
\textbf{Output}: a $\preceq$-repair for $(q,E)$
}}

\medskip

We will also study the analogous algorithmic problems for  $\preceq$-generalization and
 $\preceq$-specialization, 
 defined in the expected way. 

In all of the above problems, the input queries and examples are assumed to be
compatible in terms of their  arity.  Moreover, in our complexity analyses, we will
assume a fixed (constant) query arity $k\geq 0$. This
is in fact only necessary for some of the upper bounds in 
Section~\ref{sec:specializations}.

\section{Containment-Based Approach}
\label{sec:containment-based}

In this section, we study notions of query generalization, query specialization and query repair defined based on 
query containment. For generalization and specialization, 
it seems intuitively clear what the definition 
should be. Let $(q,E)$ be an annotated CQ. Then
\begin{itemize}
    \item a \emph{containment-based generalization} for $(q,E)$ is a CQ 
    $q'$ that fits $E$ and such that $q\subseteq q'$ and
    there is no CQ $q''$ that fits $E$ with  $q \subseteq q'' \subsetneq q'$.
    \item a \emph{containment-based specialization} for $(q,E)$ is  a CQ 
    $q'$ that fits $E$ and such that $q'\subseteq q$ and
     there is no CQ $q''$ that fits $E$ with  $q' \subsetneq q'' \subseteq q$.
\end{itemize}
It is less immediately clear what the right query containment-based definition of \emph{query repairs} should be. 
As it turns out, the above notions of query generalization and query specialization can be viewed as query generalizations and query specializations with respect to the following natural proximity pre-order (cf.~also~\cite{Satoh1988:nonmonotonic,Winslett1990:updating,Barcelo2020:static}).

\begin{defi}[Containment of Difference] 
For queries $q, q_1, q_2$, we write $q_1\preceq^{\dom}_{q} q_2$
if 
$\extension{q}\oplus \extension{q_1}\subseteq \extension{q}\oplus \extension{q_2}$, 
where $\oplus$ denotes  symmetric difference. 



\end{defi}

Recall that $\extension{q}$ denotes the set of all positive examples of a query $q$. Therefore, $\extension{q}\oplus \extension{q_i}$ denotes the set of all 
examples on which $q_i$ disagrees with $q$. Thus, 
$q_1\preceq^{\dom}_q q_2$ means that the set of examples 
on which $q$ and $q_1$ disagree is a subset of the 
set of examples on which $q$ and  $q_2$ disagree (cod
stands for ``containment of difference'').





It is easy to see that $\preceq^{\dom}$ is indeed a proximity pre-order. Moreover,
it gives rise to the intended containment-based notions of query generalization and specialization:

\begin{prop}
\label{prop:dom-generalization-specialization}
    For all annotated CQs $(q,E)$ and CQs $q'$,
    \begin{enumerate}
        \item $q'$ is a $\preceq^{\dom}$-generalization for $(q,E)$ iff $q'$ is a containment-based  generalization for~$(q,E)$ 
        \item $q'$ is a $\preceq^{\dom}$-specialization for $(q,E)$ iff $q'$ is a containment-based  specialization for $(q,E)$. 
    \end{enumerate}
\end{prop}

\begin{proof}
To start, we will argue that
\begin{enumerate}
\item[(a)] The $\preceq^{\dom}$-generalizations for $(q,E)$ are precisely the $\preceq^{\dom}$-repairs $q'$ for $(q,E)$ satisfying $q\subseteq q'$.
\item[(b)] The $\preceq^{\dom}$-specializations for $(q,E)$ are precisely the $\preceq^{\dom}$-repairs $q'$ for $(q,E)$ satisfying $q'\subseteq q$.
\end{enumerate} 
    We will only discuss (a). The argument for (b) is analogous.
    The right-to-left direction is immediate. For the left-to-right
    direction, let $q'$ be a $\preceq^{\dom}$-generalization for $(q,E)$ and
    suppose for the sake of a contradiction that $q'$ is not a 
    $\preceq^{\dom}$-repair. Then there is a CQ $q''$ that fits $E$ such that
    $\extension{q}\oplus \extension{q''}\subsetneq \extension{q}\oplus\extension{q'}$.
    Since $q\subseteq q'$, it follows that $q\subseteq q''\subsetneq q'$. This contradicts
    the fact that $q'$ is a $\preceq^{\dom}$-generalization for $(q,E)$.

The left-to-right direction of 1 and 2 is immediate from (a) and (b) above.
For the right-to-left direction, we will prove it for 1. 
If $q\subseteq q'$, then
$\extension{q}\oplus\extension{q'}=\extension{q'}\setminus\extension{q}$, and hence a
$q''\preceq^{\dom}_q q'$ holds iff
$\extension{q}\subseteq\extension{q''}\subseteq \extension{q'}$. The proof for the right-to-left direction of 2 is similar.
\end{proof}

It also follows that  $\preceq^{\dom}$-generalizations and  $\preceq^{\dom}$-specializations 
 are $\preceq^{\dom}$-repairs.
This furthermore suggests  $\preceq^{\dom}$-repairs as 
a (seemingly) natural containment-based  notion of query repair.
Next,
we will study 
$\preceq^{\dom}$-generalizations, 
$\preceq^{\dom}$-specializations, 
and $\preceq^{\dom}$-repairs.
Our main findings can be summarized as follows: \emph{$\preceq^{\dom}$-generalizations} and 
\emph{$\preceq^{\dom}$-specializations} are well-behaved notions, although the latter do not always exist (Example~\ref{ex:running-specialization}) and can be too plentiful (Example~\ref{ex:dominance-specializations-infinite}).
The associated existence, verification, and construction problems admit effective algorithms, although often of super-polynomial complexity. The more general  \emph{$\preceq^{\dom}$-repairs} exhibit counter-intuitive behavior.


\subsection{Containment-Based Query Generalizations}

The following example illustrates
$\preceq^{\dom}$-generalizations.

\begin{exa}
     Consider the schema consisting of unary relations $P,Q$. Let
    \[q(x) \colondash  P(x), Q(x)\] and 
    let $I$ be the instance that consists of the facts $P(a), Q(b)$.
    There is exactly one $\preceq^{\dom}$-generalization for $(q,E^+)$, where $E^+$ is the set of positive examples
    $\{(I,a)\}$, namely  \[q'(x)\colondash P(x), Q(y).\] This is, in fact, also the only $\preceq^{\dom}$-repair.
\end{exa}

The next  results show that there is a precise, two-way correspondence between $\preceq^{\dom}$-generalizations and most-specific fitting CQs. A  \emph{most-specific fitting CQ} for a collection of labeled examples $E$ is a fitting CQ $q$ such that 
    for every fitting CQ $q'$, $q\subseteq q'$
    \cite{pods2023:extremal}.

\begin{thm}
\label{thm:dom-generalization-most-specific}
    For all CQs $q,q'$ and collections of labeled examples $E=(E^+,E^-)$,
    \begin{enumerate}
        \item $q'$ is a $\preceq^{\dom}$-generalization for $(q,E)$ iff
         $q'$ is a most-specific fitting CQ for 
        $(E^+\cup\{e_q\},E^-)$.
        \item $q$ is a most-specific fitting CQ for $E$ iff $q$ is a $\preceq^{\dom}$-generalization for $(q_\bot, E)$,
    \end{enumerate} 
where $q_\bot$ denotes the maximally-constrained CQ over the relevant schema $\calS=\{R_1, \ldots, R_n\}$ and of the relevant arity, i.e., the CQ
$q_\bot(x,\ldots,x) \colondash R_1(x,\ldots,x), \ldots, R_n(x,\ldots, x)$.
\end{thm}

\begin{proof}
    1. Observe that a CQ $q'$ fits $e_q$ as a positive example if and only if $q\subseteq q'$, which implies that $\extension{q'}\oplus\extension{q} = \extension{q'}\setminus\extension{q}$. It follows that
a CQ is a $\preceq^{\dom}$-generalization for $(q,E)$
if and only if it is most-specific among the CQs
that fit~$(E^+\cup\{e_q\},E^-)$.

2.    Observe that $q_\bot\subseteq q$ holds for \emph{every} CQ $q$, 
   and, consequently, 
   $\extension{q_\bot}\oplus \extension{q} = \extension{q}\setminus \extension{q_\bot}$.
   It follows that $q_1\preceq_{q_\bot} q_2$ holds if and only if
   $q_1\subseteq q_2$. Therefore, 
   a query is a $\preceq^{\dom}$-generalization for $(q_\bot,E)$ 
   if and only if 
   it is a most-specific fitting CQ for $E$.
\end{proof}

As a consequence of this, we can leverage known 
results about most-specific fitting CQs. For instance,
it is known that, for every collection of labeled examples $E$, there is at most one
most-specific fitting CQ up to equivalence,
namely the CQ whose canonical example is the 
direct product of the positive examples in $E$. This implies:

\begin{cor}
    Let $(q,E)$ be an annotated CQ for which a fitting CQ $q'$ with $q\subseteq q'$ exists. Then there is, up to equivalence, 
    exactly one $\preceq^{\dom}$-generalization for $(q,E)$.
\end{cor}

\begin{proof}
    It was observed in~\cite{pods2023:extremal}
    that, for any collection of labeled examples $E$, if there exists a fitting CQ for $E$ then there is, up to equivalence, 
    precisely one most-specific fitting CQ for $E$. By the first item of Theorem~\ref{thm:dom-generalization-most-specific}, the same
    holds for $\preceq^{\dom}$-generalizations.
\end{proof}

\begin{exa}[Example~\ref{ex:running-generalization} revisited]
    Using Theorem~\ref{thm:dom-generalization-most-specific}, 
    we can verify the claim, in Example~\ref{ex:running-generalization}, that
    the CQ $q'$ expressing ``$x$ lies on a directed $R$-cycle of length 12'' is 
    the unique (up to equivalence) $\preceq^{\dom}$-generalization for $(q,E)$.
    This is true because the canonical example of $q'$ is homomorphically equivalent to the 
    direct product of the positive example $(I,a)$ and the canonical example
    $e_q$ of the input CQ.
\end{exa}

The complexity of 
various algorithmic problems pertaining to most-specific fitting CQs, as well as size bounds, were studied in~\cite{pods2023:extremal}. From these, we obtain complexity results and size bounds for $\preceq^{\dom}$-generalizations. We include here also
an analysis for the case where the input 
consists of positive examples only, which is 
particularly natural in the case of 
query-generalizations: it follows from Theorem~\ref{thm:dom-generalization-most-specific} and what is said below it
that negative examples have no
effect on generalizations except that they may render
them non-existent.  

\begin{cor}
\label{cor:complexity-dom-generalization} 
    \begin{enumerate}      
        \item \problem{$\preceq^{\dom}$-generalization-existence} is coNExpTime-complete.
        For inputs consisting of a bounded number of examples, it is coNP-complete.
        For inputs consisting of positive examples only, it is in PTime. 
        \item \problem{$\preceq^{\dom}$-generalization-verification} is NExpTime-complete, even for inputs
        consisting of positive examples only.
        It is DP-complete for a bounded number of input examples.
        \item \problem{$\preceq^{\dom}$-generalization-construction} is in ExpTime (and in PTime if the number of examples is bounded). 
        \item Let $q() \colondash R(x,x)$. There is a sequence of examples $(e_n)_{n\in\mathbb{N}}$ of size polynomial in $n$, such that (i) there is a  $\preceq^{\dom}$-generalization for $(q,E^+_n)$, where $E^+_n=\{e_1, \ldots, e_n\}$, and (ii)
        every $\preceq^{\dom}$-generalization for $(q,E^+_n)$ has size at least $2^n$. 
    \end{enumerate}
\end{cor}

\begin{proof}
    By Theorem~\ref{thm:dom-generalization-most-specific}
    Item 1, 2, and 3 follow immediately from 
    analogous complexity results for
    the existence, verification, and 
    construction problem for most-specific fitting CQs \cite{pods2023:extremal}. For 4, it suffices to 
    let $e_n$ be the directed $R$-cycle of length $p_n$, where $p_n$ is the $n$-th
    prime number. 
    It follows
    from the prime number theorem that $e_n$ is of size $O(n\ln n)$. Furthermore, the direct product $\Pi_{i=1\ldots n} (e_i)$ is homomorphically equivalent to
    a directed $R$-cycle of length $\Pi_{i=1\ldots n} (p_i)$,
    from which it follows that every most-specific fitting CQ for 
    $E_n^+$ has size at least $2^n$.
     It follows by Theorem~\ref{thm:dom-generalization-most-specific}(1)
    that every $\preceq^{\dom}$-generalization for $(q,E^+_n)$ is a most-specific fitting
    CQ for $E^+_n\cup\{e_q\}$. Since $e_q$ is maximal in the homomorphism pre-order, 
    $e_q\times\Pi_{i=1\ldots n}(e_i)$ is homomorphically equivalent to
    $\Pi_{i=1\ldots n}(e_i)$, i.e., the directed
    $R$-cycle of length $\Pi_{i=1\ldots n}(p_i)$, which is a core. Therefore,
    every $\preceq^{\dom}$-generalization for $(q,E^+_n)$ is of size at least $\Pi_{i=1\ldots n}(p_i) \geq 2^n$.

    See also \cite[Example~4.1]{SigmodRecordColumn}.
\end{proof}

\begin{rem}
\label{remark:no-size-bound-dom-repairs}
Corollary~\ref{cor:complexity-dom-generalization}(4) implies that 
the size of $\preceq^{\dom}$-repairs is in general exponential 
in the number of positive examples, and also that 
the size of $\preceq^{\dom}$-repairs for $(q,E)$ cannot be bounded by \emph{any} 
function in the size of $q$ and the size of the
smallest fitting CQ for $E$.
\end{rem}



\subsection{Containment-Based Query Specializations}
\label{sec:specializations}

The following example illustrates
$\preceq^{\dom}$-specializations.

\begin{exa}
    Consider the CQ \[q(x) \colondash  P(x)\] and let $E$ consist of
    \begin{itemize}
        \item a negative example $(I,a)$ where     $I$ 
consists only of the fact $P(a)$, and
        \item a positive example $(J,a)$ where 
        $J$ extends $I$ with the additional facts $Q(a)$ and $R(a,a)$.
    \end{itemize}
    There are two $\preceq^{\dom}$-repairs
    for $(q,E)$, 
    namely \[q'_1(x) \colondash P(x), Q(y)\] and
     \[q'_2(x) \colondash P(x), R(y,z).\] 
    These are the only two $\preceq^{\dom}$-repairs, up to
equivalence. To see this, we first observe that every $\preceq^{\dom}$-repair $q'$ of $(q,E)$ satisfies $q' \subseteq q$, and thus is a $\preceq^{\dom}$-specialization.
In fact, since $q$ fits $E^+$,
$q'\wedge q$ also  fits $E$. But $q'\not\subseteq q$
implies $\extension{q}\oplus \extension{q' \wedge q}\subsetneq \extension{q}\oplus \extension{q'}$, thus contradicting that $q'$ is a $\preceq^{\dom}$-repair of $(q,E)$.

Every CQ $q'(x)\subseteq q$ that fits $E$ contains the
atom $P(x)$. Fitting the negative example requires a
$Q$-atom or an $R$-atom. Consequently, every such $q'$
is contained in $q'_1$ or~$q'_2$.
These two fitting CQs are incomparable, so they are
precisely the maximal fitting CQs contained in $q$.
By Point~(2) of
Proposition~\ref{prop:dom-generalization-specialization},
they are thus precisely the repairs.
\end{exa}

An annotated CQ $(q,E)$ may lack  a $\preceq^{\dom}$-specialization even when a fitting CQ exists:

\begin{exa}\label{ex:dom-neg-notexists}
    Consider the Boolean CQ \[q() \colondash R(x,y),\] and let $E$ consist of
    \begin{itemize}
        \item a negative example $I$ consisting
    of facts $R(b,c), R(c,b)$, and
    \item a positive example $J$  consisting of the fact $R(a,a)$.
    \end{itemize}
    The positive example, here, is strictly speaking redundant: every CQ over the relevant schema fits it. It is included only for intuition.
    There are CQs $q'$ with $q'\subseteq q$ that fit $E$ (for instance,
    $q'() \colondash R(x,x)$ is such a query), but there does not exist
    a $\preceq^{\dom}$-specialization for $(q,E)$.
    This can be seen as follows. First, observe that every CQ $q'$ over
    the schema $\{R\}$ with a non-empty body satisfies $q'\subseteq q$, and that every CQ
    $q'$ fits the positive example $J$.
    Therefore, a CQ $q'$ with a non-empty body fits $E$ if and only if
    $e_{q'}\not\to I$. Next, a map $h$ from the variables of $q'$ to
    the $\adom(I)$ is a homomorphism from $e_{q'}$ to $I$ if and only if it assigns distinct
    values to the two variables of each atom of $q'$. In other words, $e_{q'}\to I$ holds
    if and only if the undirected graph underlying $e_{q'}$ is 2-colorable. By a well-known
    result in graph theory, this is the case if and only if $e_{q'}$ contains no undirected
    cycle of odd length. Thus, a CQ $q'$ fits $E$ if and only if its underlying undirected graph
    contains a cycle of odd length. For any such query $q'$ we can find another CQ
     $q''$ that is strictly closer to $q$, by increasing the length of the odd cycle. Formally,
     let $2n+1$ be the length of
    a shortest odd cycle in $e_{q'}$. By the sparse incomparability
    lemma~\cite{kun2013constraints}, there is a graph $G$ of girth strictly greater than
    $2n+1$ such that $G\to e_{q'}$ and $G\not\to I$. Let
    $q''$ be the canonical CQ of $G$. Then $q''$ fits $E$ (because $G\not\to I$), and
    $q''\subseteq q$, and $q'\subseteq q''$ (because $G\to e_{q'}$). Moreover, the latter
    containment is strict: a homomorphism $e_{q'}\to G$ would map the odd cycle of length
    $2n+1$ in $e_{q'}$ to a closed walk of odd length $2n+1$ in $G$, and hence $G$ would
    contain an odd cycle of length at most $2n+1$, contradicting the fact that the girth
    of $G$ is greater than $2n+1$. In summary, $q'\subsetneq q''\subseteq q$ and $q''$
    fits $E$, showing that $q'$ is not a $\preceq^{\dom}$-specialization for $(q,E)$.
\end{exa}

We next return to Example~\ref{ex:running-specialization} from the introduction.
\begin{prop}
\label{prop:no-wmg-path-example}
There is no $\preceq^{\dom}$-specialization for the annotation query $(q,E)$
from Example~\ref{ex:running-specialization}.
\end{prop}

\begin{proof}
Recall that $q(x)\colondash R(x,y),R(y,z)$, that the
negative example is $(I,a)$ with
$I=\{R(a,b),R(b,c)\}$, and that the positive example is
$(J,a)$ with $J=\{R(a,b),R(b,c),R(c,d)\}$.
It is convenient to describe the CQs that fit $E$ in
terms of \emph{level functions}: for a CQ $q'(x)$ over
the schema $\{R\}$, a homomorphism from $e_{q'}$ to
$(J,a)$ is the same thing as a map $\ell$ from the
variables of $q'$ to $\{0,1,2,3\}$ with $\ell(x)=0$ and
$\ell(v)=\ell(u)+1$ for every atom $R(u,v)$ of $q'$,
and similarly for $(I,a)$, with $\{0,1,2\}$ in place
of $\{0,1,2,3\}$.
Hence $q'$ fits $E$ if and only if it admits such a
level function $\ell$ with values in $\{0,1,2,3\}$ but
none with values in $\{0,1,2\}$.
In that case, $e_{q'}$ contains a path connecting a
variable of level 0 to a variable of level 3.

For $n\geq 0$, let $p_n$ be the CQ
\[
    p_n(x)\colondash
    R(x,y),R(y,z),\beta_1,\ldots,\beta_{2n+3},
\]
where
$\beta_1,\ldots,\beta_{2n+3}$ 
do not use the variables $x,y$ and 
form an
oriented $R$-path whose successive variables have levels
$0,1,(2,1)^n,2,3$, visualized by
$$\longrightarrow\cdot\longrightarrow\cdot(\longleftarrow\cdot\longrightarrow)^n\cdot\longrightarrow.$$
Each $p_n$ fits $E$: it admits a level function with
values in $\{0,1,2,3\}$ that maps $x$ to 0, while no
level function with values in $\{0,1,2\}$ exists,
because the second connected component of $e_{p_n}$
has height 3.
Also, $p_n\subseteq q$.
Moreover, $p_m\subseteq p_{m+1}$, because
$e_{p_{m+1}}\to e_{p_m}$: map the two initial atoms to
themselves and traverse the path component of $e_{p_m}$,
using $m$ of the $m+1$ ``zigzags'' of $e_{p_{m+1}}$ to
advance and the remaining one to move back and forth
along the same pair of atoms.
Furthermore, $p_m\subsetneq p_{m+1}$: a homomorphism
$e_{p_m}\to e_{p_{m+1}}$ would map the second connected
component to the second connected component, since
its levels range from 0 to 3 and the component containing
the answer variable has height 2.
Since homomorphisms preserve levels up to a shift,
its image would be a connected subinstance containing
the unique variables of level 0 and of level 3 of the
latter component, that is, the entire path component.
This is impossible because the path component of
$e_{p_m}$ has fewer atoms.

Now let $q'$ be any CQ that fits $E$ with
$q'\subseteq q$, let $\ell$ be as above, and let $w$
be a shortest path in $e_{q'}$ connecting a variable
of level 0 to a variable of level 3.
No variable strictly between the endpoints of $w$ has
level 0 or level 3, since otherwise a shorter such path
would exist.
The levels along $w$ therefore alternate between 1 and
2 in between, that is, the sequence of levels along $w$
is $0,1,(2,1)^m,2,3$, where $2m+3$ is the length of $w$.
In other words, the path component of $e_{p_m}$ maps
homomorphically onto $w$.
Furthermore, $q'\subseteq q$ gives a homomorphism
$e_q\to e_{q'}$.
Mapping the two initial atoms of $p_m$ according to
this homomorphism and its variable-disjoint path
component onto $w$, we obtain a homomorphism
$e_{p_m}\to e_{q'}$, that is, $q'\subseteq p_m$.

Since $q'\subseteq p_m$ and $p_m\subsetneq p_{m+1}$,
we have $q'\subseteq p_{m+1}$, and this containment is
strict, because $q'\equiv p_{m+1}$ would give
$p_{m+1}\subseteq p_m$.
As $p_{m+1}$ fits $E$ and $p_{m+1}\subseteq q$, it
follows that $q'$ is not a $\preceq^{\dom}$-specialization
for $(q,E)$.
\end{proof}

    In the previous subsection, we saw that $\preceq^{\dom}$-generalizations are closely related to 
    most-specific fitting CQs. Similarly, 
    $\preceq^{\dom}$-specializations are
 related to weakly
    most-general fitting CQs, where
    a \emph{weakly most-general fitting CQ} for a collection of labeled examples $E$ is
    a fitting $q$ such that 
    for every fitting CQ $q'$, $q\subseteq q'$ implies $q \equiv q'$
    \cite{pods2023:extremal}. In this case, however, the connection  only 
    works out cleanly for Boolean queries.\footnote{In the conference verion of this article, the characterization was erroneously stated for CQs of unrestricted arity. Related to this, also Corollary~\ref{cor:cod-spec-complexity} was stated in the conference version of queries of unrestricted arity.}

\begin{thm}
\label{thm:reduction-specialization}
    Let $(q,E)$ be an 
    annotated Boolean CQ with $E=(E^+,E^-)$, such that $q$ has no repeated answer variables.
    Then, for all Boolean CQs $q'$, the following are equivalent: 
    \begin{enumerate}
        \item $q'$ is a $\preceq^{\dom}$-specialization for $(q,E)$,
        \item $q$ fits $E^+$ and $q'$ is equivalent to $q\land q''$ for some $q''$ that is a weakly most-general fitting CQ for
        $(E^+,E'^-)$, where
        $E'^-=\{e\in E^-\mid e\in\extension{q}\}$. \end{enumerate}
    Moreover, in the direction from 1 to 2, $q''$ can  be chosen so that
    $|q''|\leq |q'|$.
\end{thm}

\begin{proof}
\emph{From 1 to 2.}
Let $q'$ be a $\preceq^{\dom}$-specialization for $(q,E)$.
Since $q'$ fits $E^+$ and $q'\subseteq q$, the query $q$
fits $E^+$.
Let $q''$ be the conjunction of those connected components
of $q'$ that do not homomorphically map to $q$.
Since $q'\subseteq q$ and every  component in $q' \setminus q''$ is
implied by $q$, we have
$
    q'\equiv q\land q''.
$
Moreover, $|q''|\leq |q'|$.

The query $q''$ fits $(E^+,E'^-)$.
It fits $E^+$ because $q'\subseteq q''$.
For any $e\in E'^-$, some component $p$ of $q'$ satisfies
$e\not\in\extension{p}$.
Since $e\in\extension{q}$, $p$ cannot
homomorphically map to $q$.
Hence $p$ is retained in $q''$, implying
$e\not\in\extension{q''}$.

To prove that $q''$ is weakly most-general fitting, let $q'''$ fit $(E^+,E'^-)$
with $q''\subseteq q'''$ and assume w.l.o.g.\ that $q$ and $q'''$ do not share any variables.
Then $q\land q'''$ fits $E$, since both conjuncts fit
$E^+$ and every $e\in E^-$ satisfies
$e\not\in\extension{q}$ or $e\not\in\extension{q'''}$.
Furthermore,
\[
    q'\equiv q\land q''
    \subseteq q\land q'''\subseteq q.
\]
Since $q'$ is a $\preceq^{\dom}$-specialization for $(q,E)$, it follows that
$q\land q'''\equiv q'$.
In particular, $q\land q'''\subseteq q''$, so every
component $p$ of $q''$ homomorphically maps to $q\land q'''$.
But since $q$ and $q'''$ do not share variables then each such component $p$ maps entirely to $q$ or to $q'''$.
By construction of $q''$, no $p$ maps to $q$, so all $p$ map to $q'''$.
Consequently, $q'''\subseteq q''$, and thus
$q'''\equiv q''$ as required.

\medskip\noindent
\emph{From 2 to 1.}
Suppose that $q$ fits $E^+$ and
$q'\equiv q\land q''$, where $q''$ is weakly most-general
fitting for $(E^+,E'^-)$ that does not share any variables with $q$.
We have to show that $q'$ is a $\preceq^{\dom}$-specialization for $(q,E)$.
It is easy to verify that $q'\subseteq q$ and $q'$ fits $E$.

Suppose, towards a contradiction, that a Boolean CQ $q'''$
fits $E$ and satisfies
\[
    q'\subsetneq q'''\subseteq q.
\]
Let $q_0$ be the conjunction of those components of $q'''$
that do not homomorphically map to~$q$.
The same  arguments as in the first
direction show that $q_0$ fits $(E^+,E'^-)$ and
$
    q'''\equiv q\land q_0.
$

The containment $q'\subseteq q'''$ gives a homomorphism
from $q'''$ to $q\land q''$.
Every component $p$ of $q''' \setminus q_0$ must map entirely to $q''$,
since it does not map to $q$.
It follows that $q''\subseteq q_0$.
Since $q''$ is weakly most-general
fitting for $(E^+,E'^-)$, we obtain $q_0\equiv q''$.
Hence
\[
    q'''\equiv q\land q_0
         \equiv q\land q''
         \equiv q',
\]
contradicting $q'\subsetneq q'''$.
\end{proof}
The following example illustrates that
Theorem~\ref{thm:reduction-specialization}
does not extend to queries of unrestricted arity.
\begin{exa}
\label{ex:non-boolean-specialization}
Let the schema $\calS$ contain a binary relation
$R$ and a unary relation~$S$, and consider the annotated unary CQ $(q,E)$ with
$
    q(x)\colondash R(x,x)$ and $   E=(E^+,E^-)$ for 
  $
    E^+=\{(I^+,a)\}$
    and
    $
    E^-=\{(I^-,a)\}$,
  where
$$
    I^+=\{R(a,a),S(a)\}
    \quad\text{ and }\quad
    I^-=\{R(a,a),S(b)\}.
$$
The query $q$ fits $E^+$. Since it also fits the negative example,
 $E'^-=E^-$ in the notation of
Theorem~\ref{thm:reduction-specialization}. Now consider
the queries
   $q'(x)\colondash R(x,x),S(x)$,
    and
$
    q''(x)\colondash S(x).
$ The latter is a weakly most-general fitting CQ for $E$.
Indeed, every safe unary CQ more general than $S(x)$ is
equivalent to $S(x)$: its body can contain only $S$-atoms,
and safety requires an atom $S(x)$.\footnote{
It is worth pointing out that the same counterexample applies when the safety requirement is dropped: in this case, $q''$ remains weakly most-general fitting, and $p$ still witnesses that $q'=q\wedge q''$ is not a $\preceq_{\mathrm{cod}}$-specialization.} 
Nevertheless, $q'=q \wedge q''$ is not a $\preceq^{\dom}$-specialization of $(q,E)$.
In fact,
\[
    p(x)\colondash R(x,x),R(x,y),S(y).
\]
fits $E$ and satisfies
$
    q'\subsetneq p\subseteq q.
$
\end{exa}
There is also a converse reduction, 
and it even works for CQs of unrestricted arity. A slight technical complication
arises from the fact that, due to the safety condition in the definition of CQs, for $k>0$ there does not exist a $k$-ary
``most-general'' CQ $q_\top$ that is contained in all $k$-ary CQs. Instead,
we consider all CQs $q(x_1, \ldots, x_k)$ whose body only contains, 
for each $i\leq k$ one atom of the form 
$R(\textbf{y},x_i,\textbf{z})$, where $R$ is any relation symbol and 
$\textbf{y},\textbf{z}$ are tuples of distinct fresh existential variables. We  call  any such CQ a \emph{minimally-constrained CQ}.  Note that, for any schema and arity,
there are only finitely many  minimally-constrained CQs, up to the renaming of variables. 

\begin{thm}
    \label{thm:wmg-as-dom-specialization}
    For every CQ $q$ and collection $E$ of labeled examples, the following are equivalent:
    \begin{enumerate}
        \item $q$ is a weakly most-general fitting CQ for $E$,
        \item $q$ is a $\preceq^{\dom}$-specialization for $(q_\top,E)$ for all minimally-constrained CQs $q_\top$  with $q\subseteq q_\top$.
    \end{enumerate}
\end{thm}

\begin{proof}
    (1 to 2:)
    Let $q$ be a weakly most-general fitting CQ for $E$, and consider any 
    minimally-constrained query $q_\top$ such that $q\subseteq q_\top$. 
    We claim that $q$ is a  $\preceq^{\dom}$-specialization for $(q_\top,E)$. Suppose otherwise. Then there exists a CQ $q'$ that fits $E$, with
    $q\subsetneq q'\subseteq q_\top$. But then $q'$ is a fitting CQ for $E$
    that is strictly more general than $q$, a contradiction.

    (2 to 1:) Let $q$ be a  $\preceq^{\dom}$-specialization for $(q_\top,E)$ for all minimally-constrained CQs $q_\top$ for which it holds that $q\subseteq q_\top$.
    We claim that $q$ is a weakly most-general fitting CQ for $E$. 
    Suppose, for the sake of a contradiction, that 
    there is a fitting CQ $q'$ such that $q\subsetneq q'$.
    Let $q_\top$ be a minimally-constrained CQ such that $q'\subseteq q_\top$.
    Such a $q_\top$ always exists and can be read off from $q'$ as
    follows. Let $y_1, \ldots,  y_k$ be the answer variables of $q'$. By the safety condition, 
    for each $i\leq k$,
    there is a relation $R_i$ and an index    $1\leq j_i\leq arity(R_i)$
    such that $y_i$ occurs in position $j_i$ in some $R_i$-atom in the body of $q'$.
    We let the answer variable of $q_\top$ be $y_1, \ldots, y_k$ and we include in the
    body of $q_\top$, for each $i\leq k$, an $R_i$-atom with $y_i$
    in position $j_i$ and fresh variables in all other positions. By construction,
     $q_\top$ is
    minimally-constrained and  $q'\subseteq q_\top$.
    By transitivity, we have that $q\subseteq q_\top$, and hence,
    by assumption, $q$ is a $\preceq^{\dom}$-specialization of~$q_\top$. 
    At the same time, we have that $q\subsetneq q'\subseteq q_\top$, 
    which gives us 
    our contradiction.
\end{proof}

For a given CQ $q(x_1, \ldots, x_k)$,
the set of all minimally-constrained CQs $q_\top$ with $q\subseteq q_\top$ can be constructed
in polynomial time (for fixed query arity).
Indeed, a minimally-constrained CQ is determined, up to renaming of
variables, by a choice, for each $i\leq k$, of a relation symbol $R_i\in\calS$ and a
position $j_i\leq \mathit{arity}(R_i)$. Hence there are at most $m^k$ minimally-constrained
CQs, where $m=\Sigma_{R\in\calS}\mathit{arity}(R)$, and they can all be enumerated in
polynomial time for fixed $k$. Moreover, since the existential variables of a
minimally-constrained CQ $q_\top$ are pairwise distinct and each of them occurs in only one
atom, $q\subseteq q_\top$ holds if and only if, for each $i\leq k$, the body of $q$
contains an $R_i$-atom in which $x_i$ occurs in position $j_i$, which can be tested in
polynomial time.
This means that the above theorem
can be viewed as a polynomial-time (Turing) reduction.

Weakly most-general fitting CQs were studied in depth in \cite{pods2023:extremal}. Based on  results from that article and the above reductions, we obtain a number of results.
We start with the following example which shows that, in contrast to the case of $\preceq^{\dom}$-generalizations, an annotated CQ may have infinitely many non-equivalent $\preceq^{\dom}$-specializations.

\begin{exa} 
   \label{ex:dominance-specializations-infinite}
    Consider the CQ \[q()\colondash R(x,y)\] and let $E$ consist of
    \begin{itemize}
        \item a positive example consisting of the facts $R(a,a), P_1(a), P_2(a)$, and 
        \item a negative example consisting of the facts $R(a,a), R(b,b), R(a,b), P_1(a), P_2(b)$. 
    \end{itemize}
    Note that the positive example is strictly speaking redundant as it belongs to $\extension{q'}$ for all CQs $q'$. It is added only for intuition.
    It follows from results in~\cite{pods2023:extremal} that 
    there are infinitely many  non-equivalent weakly most-general fitting CQs for $E$. 
    Indeed, for every $n\geq 1$, the  CQ
    \[q_n() \colondash R(x_1, x_2), \ldots, R(x_{n-1}, x_n), P_2(x_1), P_1(x_n)\] is a weakly most-general fitting CQ for $E$. It follows
    by Theorem~\ref{thm:reduction-specialization}
    that there are infinitely many 
    $\preceq^{\dom}$-specializations for $(q,E)$. 
\end{exa}

We next use the results from \cite{pods2023:extremal} to derive complexity
results for decision problems that pertain to $\preceq^{\dom}$-specializations. 

\begin{cor}~For annotated Boolean queries,
\label{cor:cod-spec-complexity}
    \begin{enumerate}
        \item $\preceq^{\dom}$-\problem{specialization-existence} is ExpTime-complete. 
        \item $\preceq^{\dom}$-\problem{specialization-verification} is in $P^{NP}_{||}$ and is DP-hard already for  a bounded number of input examples. 
        \item $\preceq^{\dom}$-\problem{specialization-construction} is in 2ExpTime.
        \item There is a sequence $(q_n,E_n)_{n\in\mathbb{N}}$, where $q_n$ and $E_n$ are of total size polynomial in $n$, such that
        (i) there is a $\preceq^{\dom}$-specialization for $(q_n,E_n)$, and (ii) 
        every $\preceq^{\dom}$-specialization for $(q_n,E_n)$ is of size at least $2^n$.
    \end{enumerate}
\end{cor}

\begin{proof} 
  Before we start, we note that, 
  in each of the three problems, 
  we may assume without loss of generality that the input annotated
  CQ $(q,E)$ is such that 
  $q$ has no repeated answer variables. 
  Indeed, consider any CQ $q(\textbf{x})$ where  $\textbf{x}=x_1\ldots x_k$ is a tuple of variables containing 
  repeated occurrences of the same variable(s).
  Let $k'<k$ be the number of distinct variables in the tuple $\textbf{x}$,
  and let $\textbf{x}'=x_{i_1}\ldots x_{i_{k'}}$, where $i_j$ denotes
  the first position in which the
  $j$-th distinct variable occurs in 
  the tuple $\textbf{x}$.
  We will denote by $\widehat{q}$  the CQ that has the same body as $q$ and whose tuple of answer variables is $\textbf{x}'$. Note that $\widehat{q}$ is, by construction, 
  repetition-free.
   Furthermore, observe that every CQ $q'$ with $q'\subseteq q$ satisfies
  $\sim_q\subseteq\sim_{q'}$, that is, the answer tuples produced by $q'$ satisfy all
  equalities that are dictated by $\textbf{x}$: indeed, if $q'\subseteq q$, then there is a
  homomorphism $h:e_q\to e_{q'}$ that maps the $i$-th answer variable of $q$ to the $i$-th
  answer variable of $q'$, and therefore $x_i=x_j$ implies that the $i$-th and the $j$-th
  answer variable of $q'$ coincide. Let us call an example $e=(I,\textbf{a})$
  \emph{compatible} if $a_i=a_j$ whenever $x_i=x_j$. It follows that (i) if some example in
  $E^+$ is not compatible, then no CQ $q'\subseteq q$ fits $E$ and hence there is no
  $\preceq^{\dom}$-specialization for $(q,E)$, and (ii) every CQ $q'\subseteq q$ fits every
  incompatible example in $E^-$. Therefore, we may reject the input immediately in case
  (i) and otherwise discard the incompatible negative examples, which does not affect the
  set of $\preceq^{\dom}$-specializations. Let $\widehat{E}$ be obtained from the remaining
  (compatible) examples by replacing each $e=(I,\textbf{a})$ by the example
  $(I,\textbf{a}')$ with $\textbf{a}'=a_{i_1}\ldots a_{i_{k'}}$, and, for a CQ $q'$, let
  $\widehat{q'}$ be obtained from $q'$ by restricting its tuple of answer variables to the
  positions $i_1, \ldots, i_{k'}$. Since all examples under consideration are compatible
  and since every $q'\subseteq q$ satisfies $\sim_q\subseteq\sim_{q'}$ (so that the omitted
  positions carry no information), the map $q'\mapsto\widehat{q'}$ is a one-to-one
  correspondence, up to equivalence and computable in polynomial time in both directions,
  between the CQs $q'\subseteq q$ that fit $E$ and the CQs $q'\subseteq\widehat{q}$ that
  fit $\widehat{E}$, and it preserves containment in both directions. Hence it restricts to
  a one-to-one correspondence between the $\preceq^{\dom}$-specializations for $(q,E)$ and
  the $\preceq^{\dom}$-specializations for $(\widehat{q},\widehat{E})$.

  Therefore, 
  we can restrict attention to the repetition-free case.

 1. It was shown in \cite{pods2023:extremal} that
    the existence problem for weakly most-general fitting CQs is ExpTime-complete.
 It follows by Theorem~\ref{thm:reduction-specialization} and the fact that fitting verification is in ExpTime that  $\preceq^{\dom}$-\problem{specialization-existence} is in ExpTime. 
 Despite Theorem~\ref{thm:wmg-as-dom-specialization}, this does not immediately imply ExpTime-hardness since the queries used in the hardness proof in \cite{pods2023:extremal} are not Boolean. However, we argue
 that the proof (of
\cite[Theorem~D.14]{pods2023:extremal}) can be adapted to the Boolean case.
Given a word $w$ for a fixed halting alternating
linear-space Turing machine $M$, it constructs unary
positive examples $(I_1,a_1),\ldots,(I_m,a_m)$ and a
single negative example $(J,b)$, both using only unary and binary relation symbols.
If $M$ accepts $w$, then
\[
    \Big(\prod_{i=1}^m I_i,(a_1,\ldots,a_m)\Big)
    \to (J,b).
\]
 If $M$ rejects $w$, the proof constructs a unary tree CQ $p(x)$
  that fits these examples and is weakly most-general fitting.
The query $p(x)$ is a core
since, when viewed as a tree, every node in it has at most  one outgoing
edge per relation symbol.

 Let $A$ be a fresh unary relation symbol. Replace each positive
  example $(I_i,a_i)$ by the Boolean example
  $I^A_i = I_i\cup\{A(a_i)\}$ and the negative example $(J,b)$ by the Boolean example
  $J^A=J\cup\{A(b)\}$. Denote the resulting collection of Boolean
  examples by $\widehat{E}$. 
If $M$ accepts~$w$, the displayed homomorphism also witnesses $\prod_i I^A_i \to J^A$.
Every Boolean CQ fitting $\widehat E^+$ homomorphically
maps to this product and hence to the negative example.
Thus no Boolean CQ fits $\widehat E$.

Now suppose that $M$ rejects $w$. Let $p'()$ be obtained from $p(x)$
  by adding $A(x)$ to its body and existentially quantifying $x$.
  Since $A$ holds only at the distinguished element of each marked
  example, $p'$ fits $\widehat{E}$. Moreover, $p'$ is a core since $p$ is.
  It remains to check that $p'$ is weakly most-general fitting.
  The proof of \cite[Theorem~D.14]{pods2023:extremal} constructs an
  auxiliary CQ $p_0(x)$ from $p$ that has two properties:
  (i)~every proper generalization of $p$ has a homomorphism to $p_0$ and (ii)~$p_0$ has a homomorphism to $(J,b)$, with $x$ mapped to $b$.
  The construction in
  \cite[Definition~B.14]{pods2023:extremal}, applied to $p'$, gives
  precisely the Boolean CQ $p'_0$ obtained by adding $A(x)$ to $p_0$ and existentially quantifying $x$.
  Consequently and since $p'$ is a core, $p'_0$ maps to
  $J^A$ and $p'_0$ is contained in every
  proper generalization of $p'$. Hence every proper
  generalization of $p'$ fits the negative example in $\widehat{E}$,
  as required.

  We have shown that $\widehat{E}$ admits a weakly most-general fitting
  Boolean CQ if and only if $M$ rejects $w$. Since ExpTime is closed
  under complement, this proves ExpTime-hardness of weakly most-general fitting existence for Boolean CQs. It remains to apply Theorem~\ref{thm:wmg-as-dom-specialization}.

  2. We recall that the complexity class $P^{NP}_{||}$ admits 
many equivalent definitions, including
as the class of problems solvable in polynomial time using a single parallel round of calls to an NP-oracle (meaning that the algorithm may make polynomially many calls to an NP-oracle but these calls must be independent of each other)~\cite{Buss1991:truth}.
  It was shown in \cite{pods2023:extremal} that
  the verification problem for weakly most-general fitting CQs is NP-complete, and that NP-hardness holds already for
  a fixed collection $E$ of negative examples. Theorem~\ref{thm:reduction-specialization} therefore implies the following
  $P^{NP}_{||}$-algorithm for 
  $\preceq^{\dom}$-\problem{specialization-verification}: we first use a parallel round of NP-oracle calls to compute $E'^-$  and to check that $q$ fits $E^+$. Next, we perform one more NP-oracle call to guess the
  existence of a polynomial-sized CQ $q''$ that is a weakly most-general fitting CQ for $(E^+,E'^-)$, together with 
  homomorphisms between $q\land q''$ and $q'$. Observe
 that this algorithm uses two rounds of parallel NP-oracle calls,
 while the definition of $P^{NP}_{||}$ we mentioned earlier allows only a 
 single round of parallel NP-oracle calls.
 This is no problem, because it was shown in \cite{Buss1991:truth} that 
  a constant number of rounds of parallel calls is no more powerful than a single round. 
  For the lower bound, it suffices to observe that for any annotated
  CQ $(q,E)$, $q$ is a $\preceq^{\dom}$-specialization for $(q,E)$ if
  and only if $q$ fits $E$ (this actually holds for any proximity pre-order).
  Since it is DP-hard to decide whether a given CQ fits a  set of labeled data examples \cite{pods2023:extremal}, even with a single positive and a single negative example, DP-hardness of  $\preceq^{\dom}$-\problem{specialization-verification} follows.
%

  3. In \cite{pods2023:extremal}, 
  an algorithm was given, based on tree automata, for constructing a  weakly most-general fitting CQ (when it exists) in 2ExpTime. It follows by Theorem~\ref{thm:reduction-specialization} that a
  $\preceq^{\dom}$-specialization
  (when it exists) can be constructed in 2ExpTime.

  4. Let $E_n$ be the collection of labeled examples provided
  by~\cite[Theorem 3.26]{pods2023:extremal} (with $n+1$ in place of $n$). It is of size
  polynomial in $n$, it uses only unary and binary relation symbols, there is a unique
  fitting CQ $q^*$ for $E_n$ up to equivalence, and every CQ that fits $E_n$ has at least
  $2^{n+1}$ variables. Being the unique fitting CQ for $E_n$, $q^*$ is in particular a
  weakly most-general fitting CQ for $E_n$. Let $q_n$ be a minimally-constrained CQ with
  $q^*\subseteq q_n$; such a CQ exists and can be read off from $q^*$ as described in the
  proof of Theorem~\ref{thm:wmg-as-dom-specialization}, and it consists of a single atom
  for each answer variable, so that $q_n$ and $E_n$ are of total size polynomial in $n$.
  By Theorem~\ref{thm:wmg-as-dom-specialization}, $q^*$ is then a
  $\preceq^{\dom}$-specialization for $(q_n,E_n)$, which proves (i). For (ii), observe that
  every $\preceq^{\dom}$-specialization for $(q_n,E_n)$ is, by definition, a CQ that fits
  $E_n$, and hence has at least $2^{n+1}$ variables. Since the schema of $E_n$ contains
  only unary and binary relation symbols, each atom of a CQ over that schema contains at
  most two variables, and therefore a CQ with at least $2^{n+1}$ variables has at least
  $2^n$ atoms.
\end{proof}
We consider it likely that all complexities listed 
in Corollary~\ref{cor:cod-spec-complexity} also hold for
queries of unrestricted arity. One way to show that would be to reuse and adapt the (somewhat introcate) algorithms from \cite{pods2023:extremal}. We do not know what happens in case 1 and 3 if the number of input examples is bounded. 

\subsection{Containment-Based Query Repairs}

We now move to the general setting of
\emph{$\preceq^{\dom}$-repairs} $q'$ for an annotated CQ $(q,E)$
where we no longer require
that $q'\subseteq q$ or $q\subseteq q'$.
Our first result on $\preceq^{\dom}$-repairs provides a reduction to the case
with only positive examples or only 
negative examples.

\begin{prop}
Let $(q,E)$ be an annotated CQ with $E=(E^+,E^-)$. Then a CQ $q'$ is a $\preceq^{\dom}$-repair for $(q,E)$ if and only if one of following holds:
\begin{enumerate}
    \item[(a)] $q'$ is
a $\preceq^{\dom}$-repair for $(q,E^+)$ and $q'$ fits $E^-$, or 
\item[(b)] $q'$ is 
 a $\preceq^{\dom}$-repair for $(q,\widehat{E^-})$ and $q'$ fits $E^+$.
\end{enumerate}
 where $\widehat{E^-} = \{e\times\Pi_{e'\in E^+}(e')\mid e\in E^-\}$ if
 $E^+\neq\emptyset$ and  $\widehat{E^-} = E^-$ otherwise.
\end{prop}
In (a) and (b), the annotated CQs $(q,E^+)$ and
$(q,\widehat{E^-})$ are considered over the same schema
as $(q,E)$, even if some relation symbols in $(q,E)$
do not occur.
\begin{proof}
    If $E^+=\emptyset$, the proposition holds trivially.
    Therefore, assume $E^+\neq\emptyset$.
    Let $q'$ be any $\preceq^{\dom}$-repair for $(q,E)$.
    We distinguish two cases. 
    
    The first case
    is where $q$ fits $E^+$. We claim that,
    in this case, (b) holds. Indeed, 
    from the fact that $q'$ fits $E$, it 
    follows by Lemma~\ref{lem:products}
    that $q'$  
    fits $\widehat{E^-}$ (as negative examples). Next,    assume for a contradiction that there
    is a CQ
    $q''$ that fits $\widehat{E^-}$ (as negative examples) with
    $q''\prec^{\dom}_q q'$. 
    
    Since $q$ and $q'$ both fit $E^+$
    and $q''\preceq^{\dom}_q q'$, 
    it follows that $q''$ fits $E^+$.
    Therefore, by Lemma~\ref{lem:products}, $q''$ also fits
    the product example $\Pi_{e'\in E^+}(e')$.
    It follows,  using~Lemma~\ref{lem:products} again, that  $q''$ fits each 
    example $e\in E^-$. In conclusion,
    $q''$ fits $E$, which contradicts the
    assumption that $q'$ is a $\preceq^{\dom}$-repair for $(q,E)$.

    The second case is where $q$ \emph{does not} fit $E^+$. We claim that, in this case, (a) holds. Assume for a contradiction that there
    is a CQ
    $q''$ that fits $E^+$ with
    $q''\prec^{\dom}_q q'$. 
    Since $q$ fails to fit $E^+$, it 
    follows by~Lemma~\ref{lem:products}  that  $q$ fits $\widehat{E^-}$ (as negative examples). Since both $q$ and $q'$
    fit $\widehat{E^-}$, and 
    $q''\preceq^{\dom}_q q'$, it follows
    that $q''$ fits $\widehat{E^-}$ as well (as negative examples).
    Since $q''$ fits $E^+$ it follows,
    using~Lemma~\ref{lem:products} again, that $q''$ fits $E^-$.
In conclusion,
    $q''$ fits $E$, which contradicts the
    assumption that $q'$ is a $\preceq^{\dom}$-repair for $(q,E)$.    

   For the converse direction, assume first that (a) holds, that is,
    $q'$ is a $\preceq^{\dom}$-repair for $(q,E^+)$ and $q'$ fits $E^-$. Then $q'$ fits
    $E$. Moreover, if a CQ $q''$ with $q''\prec^{\dom}_q q'$ were to fit $E$, then, in
    particular, $q''$ would fit $E^+$, contradicting the assumption that $q'$ is a
    $\preceq^{\dom}$-repair for $(q,E^+)$. Hence, $q'$ is a $\preceq^{\dom}$-repair for
    $(q,E)$.

    Assume next that (b) holds, that is, $q'$ is a
    $\preceq^{\dom}$-repair for $(q,\widehat{E^-})$ and $q'$ fits $E^+$. We first argue
    that $q'$ fits $E$. Since $q'$ fits $E^+$, Lemma~\ref{lem:products} implies that $q'$
    fits the product example $\Pi_{e'\in E^+}(e')$ as a positive example. If $q'$ failed
    to fit some $e\in E^-$, that is, if $e\in\extension{q'}$, then it would follow by
    Lemma~\ref{lem:products} that $e\times\Pi_{e'\in E^+}(e')\in\extension{q'}$,
    contradicting the fact that $q'$ fits $\widehat{E^-}$ as negative examples. Hence
    $q'$ fits $E$. Finally, assume for the sake of a contradiction that some CQ $q''$ with
    $q''\prec^{\dom}_q q'$ fits $E$. Then $e\not\in\extension{q''}$ for each $e\in E^-$
    and therefore, by Lemma~\ref{lem:products},
    $e\times\Pi_{e'\in E^+}(e')\not\in\extension{q''}$. In other words, $q''$ fits
    $\widehat{E^-}$, contradicting the assumption that $q'$ is a $\preceq^{\dom}$-repair
    for $(q,\widehat{E^-})$. Hence, $q'$ is a $\preceq^{\dom}$-repair for $(q,E)$.
    \end{proof}

This is promising, as
one might hope that case (a) above could
be reduced to a statement about $\preceq^{\dom}$-generalizations, since we are repairing with respect to a set of positive examples, and
likewise for case (b) and $\preceq^{\dom}$-specializations. For case (b) this approach indeed works:

\begin{prop}
\label{prop:negativespecialization}
  For all annotated CQs $(q,E)$, if $E$ consists of  negative examples only, then every $\preceq^{\dom}$-repair for $(q,E)$ is a $\preceq^{\dom}$-specialization for $(q,E)$.
\end{prop}

\begin{proof}
    Let $q'$ be a $\preceq^{\dom}$-repair for $(q,E)$ and suppose for the sake of a contradiction that $q'\not\subseteq q$
    (in other words, $\extension{q'}\setminus\extension{q}\neq\emptyset$). 
    Let $q''=q'\land q$. 
    Clearly, $q''$ fits $E$. 
    Furthermore, it is easy to see
    that $\extension{q''}\setminus \extension{q}=\emptyset$ and $(\extension{q}\setminus\extension{q''})= (\extension{q}\setminus \extension{q'})$. This shows that 
    $q''\preceq^{\dom}_q q'$. Since $\extension{q'}\setminus\extension{q}\neq\emptyset$, we in fact have that $q''\prec^{\dom}_q q'$,
    contradicting the assumption that $q'$ is a $\preceq^{\dom}$-repair for $(q,E)$.
\end{proof}

\begin{figure}

\begin{tikzpicture}[level distance=30pt, sibling distance=25pt]
  \draw (0,1.5) ellipse (2cm and 3.5cm);
  \node {$~~~~~~~~~~~~~~~~~~~ \bullet ~~~\Pi(E^+)\times e_q$} [grow'=up]
    child[opacity=0] {node {~}}
    child[-stealth] {
        node {$e_{q'}$}
        child {
            node {$~~~~~~~~~~~~ \bullet ~~~ \Pi(E^+)$}
            child[thin] {node {$e_1$}}
            child[thin] {node {$e_2$}}
            child[thin] {node {$\ldots$}}
            child[thin] {node {$e_n$}}
        }
    }
    child[-stealth,thin] {node {$e_q$}};
\end{tikzpicture}
\caption{Picture of the condition in Theorem~\ref{thm:dom-repairs-positive-characterization}.}

\label{fig:dom-repairs-positive}
\end{figure}

The same, however, does \emph{not} hold for 
case (a), as $\preceq^{\dom}$-repairs with respect to positive examples are
not necessarily $\preceq^{\dom}$-generalizations: 

\begin{exa}\label{ex:dominance-problematic}
    Consider a schema consisting of  unary relations $P,Q,R$. Let 
    \[q(x) \colondash P(x), Q(y)\] and
    let $I$ be the instance  consisting of the facts $P(a), R(b)$.
    There are, up to equivalence,  two $\preceq^{\dom}$-repairs
    for $(q,E^+)$ where $E^+=\{(I,a)\}$, namely the queries
\[q'_1(x) \colondash  P(x)\] and \[q'_2(x) \colondash  P(x), R(y).\]
    Of these, only the first is a $\preceq^{\dom}$-generalization. 
    It seems counter-intuitive that $q'_2$ is a $\preceq^{\dom}$-repair,
    since $q'_1$ is, intuitively, closer to $q$. However, there are 
    instances on which $q'_2$ agrees with $q$ but $q'_1$ does not. 
    An example is the instance
    consisting only of the fact $P(a)$.

\end{exa}

The following
result characterizes the
$\preceq^{\dom}$-repairs for a CQ and 
set of positive examples.

\begin{thm}
\label{thm:dom-repairs-positive-characterization}
Let $(q,E^+)$ be an annotated CQ where $E^+$ consists only of positive examples. Then a CQ $q'$  is a $\preceq^{\dom}$-repair for $(q,E^+)$ if and only if 
one of the following holds: 
    \begin{enumerate}
        \item $q$ fits $E^+$ and $q'$ is equivalent to $q$, or
        \item
     $q$ does not fit $E^+$, $q'$ fits $E^+$, and $(\Pi(E^+) \times e_{q\land q'}) \to e_{q'}$.
    \end{enumerate}
    In the case of Boolean CQs, item 2 can be replaced by
    \begin{enumerate}
        \item[(2')]
     $q$ does not fit $E^+$, and $(\Pi(E^+) \times e_{q}) \to e_{q'} \to \Pi(E^+)$
     (cf.~Figure~\ref{fig:dom-repairs-positive}).
    \end{enumerate}
\end{thm}

\begin{proof}
        \emph{From left to right:} It follows immediately from the definition of $\preceq^{\dom}$-repairs that if $q$ fits $E^+$, then every $\preceq^{\dom}$-repair of $q$ with respect to $E^+$ must be equivalent to $q$. Therefore, we only have to consider the case where
    $q$ does not fit $E^+$.
    If $q'$ is an $\preceq^{\dom}$-repair of $q$ with respect to $E^+$, then in particular $q'$ fits $E^+$. 
    It remains to show that $\Pi(E^+) \times e_{q\land q'}\to e_{q'}$.    
Take the canonical CQ $q'''$ of $\Pi(E^+)\times (e_{q\land q'})$
(that is, $q'''$ is the CQ that has as its 
variables the values occurring in this
example, and as atoms the facts in the example. Note that $q'''$ does indeed satisfy the safety condition, as 
can be seen from the fact that
$q'$ does and $e_{q'}\to\Pi(E^+)\times e_{q\land q'}$).
It follows
     from the construction that $q'''$ 
     fits $E^+$, and that every positive example for $q$ that is a
     positive example for $q'$ is also a positive example for $q'''$,
     while each negative example for $q$ that is a negative example
     for $q'$ is also a negative example for $q'''$. Hence,
     $q'''\preceq_q q'$. Since we have assumed that $q'$ is an $\preceq^{\dom}$-repair, it follows that $q'''$ is equivalent to $q'$, and hence
     $\Pi(E^+)\times (e_{q\land q'})\to e_{q'}$.


     \emph{From right to left:}
     If $q$ fits $E^+$ and $q'$ is equivalent to $q$,
     it follows immediately  that $q'$ is a $\preceq^{\dom}$-repair for $(q,E^+)$. Hence, we only have to consider the case
     where $q$ does not fit $E^+$.
     Assume that $\Pi(E^+)\times e_{q\land q'}\to e_{q'}$ and $q'$ fits $E^+$, i.e., $e_{q'}\to \Pi(E^+)$.      
     It follows from Lemma~\ref{lem:products} that, then, $q'$ fits $E^+$.
     We  claim that $q'$ is an $\preceq^{\dom}$-repair for $(q,E^+)$.
  Consider any CQ $q''$ that fits $E^+$
     and suppose that 
     $q''\preceq_q q'$. We will show that 
     $q''$ is equivalent to $q'$. 
     Consider the canonical example $e_{q''}$. This is a positive example for $q''$
     and a negative example for $q$ (because, otherwise, by composition of homomorphisms, we would have that $q$ fits $E^+$). In other words, $e_{q''}\in\extension{q''}\setminus\extension{q}$.
     Therefore, since 
     $q''\preceq_q q'$, it must also be a positive example for $q'$. 
     This shows that $e_{q'} \to e_{q''}$. 
     For the converse direction, observe that
     $e_{q''}\to \Pi(E^+)\times e_{q\land q'}$ 
     (this follows from the fact $e_{q\land q'}$ is a 
     positive example for both $q$ and $q'$, and hence,
     since 
     $q''\preceq_q q'$, is also a positive example for $q''$).  
     Since $\Pi(E^+)\times e_{q\land q'}\to e_{q'}$, by transitivity that 
     $e_{q''}\to e_{q'}$. 

     \emph{Equivalence of 2 and 2' in the Boolean case:} By Lemma~\ref{lem:products},
     since $e_{q}$ is a subinstance of $e_{q\land q'}$, it is clear that 2 implies 2'. Note that, by Lemma~\ref{lem:products},
     whenever $q'$ fits $E$ we have $e_{q'}\to \Pi(E^+)$. For the 
     converse direction, it follows from $e_{q'}\to \Pi(E^+)$ that $q'$ fits $E^+$. Furthermore,
     we have
     $\Pi(E^+)\times e_{q}\to e_{q'}$ by assumption,  as well as    $\Pi(E^+)\times e_{q'}\to e_{q'}$ by 
     Lemma~\ref{lem:products}. Therefore,     
     by Lemma~\ref{lem:distributivity-boolean}, 
     $\Pi(E^+)\times e_{q\land q'}\to e_{q'}$.
\end{proof}

\begin{exa}[Example~\ref{ex:running-generalization} revisited]
    Using Theorem~\ref{thm:dom-repairs-positive-characterization}, 
    one can verify the claim that
    the CQ expressing ``$x$ lies on a directed $R$-cycle of length 3'' is
    a $\preceq^{\dom}$-repair for $(q,E)$, and the same holds for the
    CQ expressing ``$x$ lies on a directed $R$-cycle of length 6'' and for the CQ
    expressing ``$x$ lies on a directed $R$-cycle of length 12''.
    Let $C_m$ denote the (pointed) canonical example of the CQ ``$x$ lies
    on a directed $R$-cycle of length $m$''. We use two well-known facts: $C_m\to C_n$ if
    and only if $n$ divides $m$; and $C_m\times C_n$ is homomorphically equivalent to 
     $C_{\mathit{lcm}(m,n)}$. In particular, $C_3\times C_4$ is homomorphically equivalent 
     to $C_{12}$. In the present example, $e_q=C_4$ and $\Pi(E^+)=C_3$. First, $q$ does not fit
    $E^+$, because $C_4\not\to C_3$. Next, let $q'$ be the CQ
    with canonical example $C_m$, for some $m\in\{3,6,12\}$. Then $q'$ fits $E^+$, because
    $C_m\to C_3$. Moreover, $e_{q\land q'}$ is the instance
    obtained from $C_4$ and $C_m$ by identifying their distinguished elements, and hence
    each connected component of $\Pi(E^+)\times e_{q\land q'}=C_3\times e_{q\land q'}$ is
    either a directed cycle of length $\mathit{lcm}(3,4)=12$, or a directed cycle of length
    $\mathit{lcm}(3,m)=m$, or the union of two such cycles sharing a single element. Each of these components maps
    homomorphically to $C_m$, and, since a homomorphism between directed cycles may be
    chosen so as to send any given element to any given element, these homomorphisms can
    be chosen so that the distinguished element of $\Pi(E^+)\times e_{q\land q'}$ is sent
    to the distinguished element of $C_m$. This shows that $\Pi(E^+)\times e_{q\land
    q'}\to e_{q'}$, so that condition 2 of
    Theorem~\ref{thm:dom-repairs-positive-characterization} is met and $q'$ is indeed a
    $\preceq^{\dom}$-repair for $(q,E)$.
\end{exa}

Using Theorem~\ref{thm:dom-repairs-positive-characterization}, we can show that there can be infinitely many $\preceq^{\dom}$-repairs for a CQ and set of positive examples.

\begin{exa}
\label{ex:dom-pos-inf-repairs}
    Over a schema consisting of a
    unary relation $P$ and a binary relation $R$, consider the unary CQ 
    \[q(x) \colondash P(x)\]
    and the set of positive examples $E^+=\{(I,b)\}$, where $I$ consists of the 
    fact $R(b,b)$. It follows
    from Theorem~\ref{thm:dom-repairs-positive-characterization} that
    \emph{every} CQ $q'$ that fits $E^+$
    is a $\preceq^{\dom}$-repair: a CQ $q'$ that fits $E^+$ contains no $P$-atom (as $I$ contains no $P$-fact),
    and therefore $(I,b)\times e_{q\land q'}$ is isomorphic to $e_{q'}$. There are infinitely
many pairwise non-equivalent such CQs, for instance the CQs expressing ``$x$ lies on a directed $R$-cycle of length $p$'' for prime numbers $p$. 
\end{exa}

The next result shows another connection between $\preceq^{\dom}$-repairs
and $\preceq^{\dom}$-generalizations:

\begin{prop}
The $\preceq^{\dom}$-generalizations for 
an annotated CQ $(q,E)$ are precisely the
$\preceq^{\dom}$-repairs for $(q,E')$, 
where $E'$ extends $E$ with the positive example $e_q$.
\end{prop}

\begin{proof}
It suffices to observe that a CQ $q'$ fits $e_q$ if
and only if $q\subseteq q'$.
\end{proof}

Example~\ref{ex:dominance-problematic} 
and Example~\ref{ex:dom-pos-inf-repairs} show that
$\prec^{\dom}$-repairs can behave counterintuitively. 
Various algorithmic results regarding the verification, existence and construction of 
$\preceq^{\dom}$-repairs can be derived from the above characterizations and reductions, but we  refrain from stating them here as they seem of little value given the problematic behavior of $\preceq^{\dom}$-repairs.

\section{Query Repairs Based on Distance Metrics}
\label{sec:distance-based}

In this section, we study proximity pre-orders 
based on distance metrics. In particular, we 
 propose and study a variant of \emph{edit distance} for CQs. We also study proximity
pre-orders based on several
other natural distance metrics. Our main findings can be summarized as follows:
edit distance, suitably defined, yields a proximity pre-order that avoids some of the problems of $\preceq^{\dom}$. We also show that other natural distance metrics induce proximity pre-orders that are less well behaved.

\begin{defi}[Semantic distance metric]
\label{def:semdistmet}
A \emph{semantic distance metric for CQs} is a function $dist(\cdot,\cdot)$
from pairs of CQs (of the same arity) to non-negative real numbers, satisfying:
\begin{enumerate}
\item $dist(q_1,q_2) = dist(q_2,q_1)$,
\item $dist(q_1,q_2) = 0$ iff $q_1$ and $q_2$ are equivalent,
\item $dist(q_1,q_2) \leq dist(q_1,q_3)+dist(q_3,q_2)$ (the triangle inequality).
\end{enumerate}
If all the conditions are met except for the 
 \emph{only-if} direction of 2,  we say that
 $dist$ is a \emph{weak semantic distance metric for CQs}.
\end{defi}

One can think of a semantic distance metric for CQs as a distance metric (in the standard sense) on the equivalence classes of CQs. Every semantic distance metric, and in fact every weak semantic distance metric, induces a pre-order.



\begin{defi}[Pre-order induced by a semantic distance metric]
Let $dist$ be a weak semantic distance metric for CQs.
We define  $\preceq^{dist}$ as follows: $q'\preceq^{dist}_q q''$ iff
$dist(q,q')\leq dist(q,q'')$. 
%
\end{defi}

%

\begin{prop}
\label{prop:dist-order-properties}
    Let $dist$ be a weak semantic distance metric for CQs. Then $\preceq^{dist}$ is a proximity pre-order. 
\end{prop}

\begin{proof}
    Conservativeness follows
from non-negativity and the if direction of property~2 of
Definition~\ref{def:semdistmet}.
For syntax independence, we first observe that replacing either
argument of $dist$ by an equivalent query leaves its value
unchanged. Indeed, suppose that $q\equiv q'$ and $p\equiv p'$.
By the if direction of property~2 and two applications of the
triangle inequality,
\[
    dist(q,p)
    \leq dist(q,q')+dist(q',p')+dist(p',p)
    = dist(q',p').
\]
Interchanging the primed and unprimed queries gives the reverse
inequality, so $dist(q,p)=dist(q',p')$. Consequently, whenever
$q_i\equiv q_i'$ for $i=1,2,3$, we have
\[
\begin{aligned}
    q_1\preceq^{dist}_{q_2}q_3
    &\iff dist(q_2,q_1)\leq dist(q_2,q_3)\\
    &\iff dist(q_2',q_1')\leq dist(q_2',q_3')\\
    &\iff q_1'\preceq^{dist}_{q_2'}q_3'.
\end{aligned}
\]
This establishes syntax independence.
\end{proof}

%
%
%

We  study $\preceq^{dist}$-repairs
for several distance metrics.
Besides the algorithmic problems of \emph{repair existence, verification and construction}, 
we also consider the following natural fitting problem that is closely related to query repairs based on distance metrics:

\medskip

\par\noindent\fbox{\parbox{.98\columnwidth}{
\problem{$dist$-bounded fitting existence} 
\smallskip

\textbf{Input}:  an annotated CQ $(q,E)$ and  a distance bound $d \geq 0$.

\textbf{Output}: \emph{Yes} if there is a CQ that fits~$E$ such that $dist(q,q') \leq d$, \emph{No} otherwise
}}

\subsection{Edit Distance}

A naive definition of the \emph{edit distance}  of two CQs $q,q'$ would be the number of 
atoms that appear in one of the two CQs but not in the other, that is, $|I_{q} \oplus I_{q'}|$.
It is easy to see that 
this is not a semantic distance metric: it violates the if direction of condition~(2) of Definition~\ref{def:semdistmet} because simple syntactic changes such as renaming a quantified variable, which do not affect the semantics of the query, affect the edit distance. This can be fixed, however, by (i) taking homomorphism cores 
(i.e., minimizing the CQs), 
and (ii) working modulo bijective variable renamings.

This leads to the 
following definition. For 
simplicity, in this section we  restrict attention  to CQs whose
sequence of answer variables
is repetition-free (a restriction that could be lifted at the expense of a more intricate definition of edit distance, cf.~Remark~\ref{rem:finer-edit-distance}).

\begin{defi}[Edit distance for CQs]
\label{def:editdist}
Given CQs $q_1(x_1, \ldots, x_k)$ and $q_2(y_1, \ldots, y_k)$, 
\[\editdist(q_1,q_2) = \!\!\!\!
\min_{\text{\begin{tabular}{c}$\rho$ a bijective variable renaming  \\ with $\rho(y_i)=x_i$ for $i=1\ldots k$\end{tabular}}} \!\!\!\!|\core(e_{q_1})\oplus \rho(\core(e_{q_2}))|
\]
where $e_1\oplus e_2$ denotes the set of facts occurring in example $e_1$ and not in $e_2$ or vice versa.
\end{defi}


\begin{exa}
 Consider the Boolean CQs
   $$
   \begin{array}{rcl}
      q_1() &\colondash& R(x_1,x_2),R(x_1,x_3),R(x_2,x_4),R(x_3,x_4) \\[1mm]
      q_2() &\colondash& R(x_1,x_2),R(x_1,x_3),R(x_2,x_4),R(x_3,x_4), A(x_2), B(x_3)
   \end{array}
$$
Note that $q_1$ is not a core, but $q_2$ is. The core of $q_1$ is obtained by
dropping the second and fourth atom. 
    Thus $\editdist(q_1,q_2)=4$. 
    Thus, perhaps surprisingly, $\editdist(q_1,q_2)$
    can be larger than the naive edit distance of $q_1$ and $q_2$ (which, in this case, is 2).
\end{exa}

\begin{prop}
    \label{prop:edit-distance-metric}
    $\editdist$ is a semantic distance metric.
\end{prop}
\begin{proof}
For all CQs $q_1(x_1,\ldots,x_k)$ and $q_2(y_1,\ldots,y_k)$,
$\editdist(q_1,q_2)=0$ if and only if $q_1$ and $q_2$ are equivalent. Indeed,
$q_1,q_2$ are equivalent CQs if and only if $e_{q_1}$ and $e_{q_2}$
are homomorphically equivalent, which holds if and only if
$\core(e_{q_1})$ and $\core(e_{q_2})$ are isomorphic, meaning that
there is a bijective variable renaming $\rho$ with $\rho(y_i)=x_i$
such that $\core(e_{q_1})=\rho(\core(e_{q_2}))$.

To show that $\editdist(q_1,q_2)=\editdist(q_2,q_1)$, it is enough
to show that $\editdist(q_2,q_1)\leq\editdist(q_1,q_2)$. Suppose that
$\editdist(q_1,q_2)=n$. Then there is a bijective variable renaming
$\rho$ with $\rho(y_i)=x_i$ such that
\[
|\core(e_{q_1})\oplus\rho(\core(e_{q_2}))|=n.
\]
Applying $\rho^{-1}$ to both arguments gives
\[
|\core(e_{q_2})\oplus\rho^{-1}(\core(e_{q_1}))|=n.
\]
Since $\rho^{-1}(x_i)=y_i$, this means that
$\editdist(q_2,q_1)\leq n$.

Finally, for the triangle inequality, suppose that
$\editdist(q_1,q_2)=n$ and $\editdist(q_2,q_3)=m$, and let
$z_1,\ldots,z_k$ be the answer variables of $q_3$. Then there are
bijective renamings $\rho,\rho'$ with $\rho(y_i)=x_i$ and
$\rho'(z_i)=y_i$ such that
\[
|\core(e_{q_1})\oplus\rho(\core(e_{q_2}))|=n
\quad\text{and}\quad
|\core(e_{q_2})\oplus\rho'(\core(e_{q_3}))|=m.
\]
Let $\rho''$ be the composition of $\rho'$ followed by $\rho$.
Then $\rho''(z_i)=x_i$ and
\[
|\rho(\core(e_{q_2}))\oplus\rho''(\core(e_{q_3}))|
=
|\core(e_{q_2})\oplus\rho'(\core(e_{q_3}))|
=m.
\]
Using the triangle inequality for symmetric difference, we obtain
\[
\begin{split}
|\core(e_{q_1})\oplus\rho''(\core(e_{q_3}))|
&\leq
|\core(e_{q_1})\oplus\rho(\core(e_{q_2}))|\\
&\quad+
|\rho(\core(e_{q_2}))\oplus\rho''(\core(e_{q_3}))|\\
&=n+m,
\end{split}
\]
and therefore $\editdist(q_1,q_3)\leq n+m$.
\end{proof}


We next take a look at the complexity of computing edit distance.

\begin{prop}
\label{prop:testingeditdist}
    Testing whether $\editdist(q_1,q_2)\leq n$
    (on input $q_1,q_2,n$) is DP-hard and in $\Sigma^p_2$. When restricted to input queries that are cores, it is NP-complete. 
\end{prop}

\begin{proof}
  For the DP lower bound, we give a polynomial time reduction from the following
  problem: given a triple of Boolean CQs $(q_1,q_2,q_3)$ with $q_3$ connected, decide
  whether $q_1$ and $q_2$ are equivalent and $q_3$ is a core. Note that this problem
  is DP-hard since CQ equivalence is NP-hard and deciding whether a connected Boolean 
  CQ is a core is coNP-hard. The latter was proved in  \cite{HN92} for undirected graphs
  (with the connectedness condition being only implicit), but then clearly also
  applies to directed graphs and Boolean CQs.

 Let $(q_1,q_2,q_3)$ be an input to the above problem.  For $i \in \{1,2\}$, let
 the CQ $q'_i$ be obtained from $q_i$ by adding the atom $A(x)$ for every
 variable $x$ of $q_i$, where $A$ is a fresh unary relation symbol. Moreover, let $q'_3$ be
 obtained from $q_3$ by adding the atom $B(x)$ for every variable $x$ of $q_3$, where $B$
 is a fresh unary relation symbol, and let $q''_3$ be obtained from $q'_3$ by adding, in
 addition, the atom $B_x(x)$ for every variable $x$ of $q_3$, where the $B_x$ are pairwise
 distinct fresh unary relation symbols.
 Let $n$ be the number of variables in $q_3$ and let $p_1$ be the disjoint union of $q'_1$ and~$q'_3$, and $p_2$ the disjoint union of $q'_2$ and $q''_3$.

 Intuitively, the markings $A$ and $B$ keep the two halves of $p_1$ and
 $p_2$ apart, so that no homomorphism maps a variable of one half to a variable of the
 other and cores and comparisons can be computed half by half, while the per-variable
 markings $B_x$ make $q''_3$ a core and force the $n$ atoms $B_x(x)$, which $p_1$ lacks
 entirely, to be added on $n$ distinct variables. The budget of $n$ edits is thus used up
 completely, and it suffices only if, in addition, the $q_3$-half of $\core(p_1)$ still has
 $n$ variables (that is, $q_3$ is a core) and the $q_i$-halves already agree (that is,
 $q_1$ and $q_2$ are equivalent).
  \\[2mm]
  {\bf Claim}. $(q_1,q_2,q_3)$ is a `yes' instance if and only if $\editdist(p_1,p_2)\leq n$.
  \\[2mm]
 Write $\Delta=\{B_x(x)\mid x$ a variable of $q_3\}$, so that
  $|\Delta|=n$ and $q''_3=q'_3\cup\Delta$. We first record two consequences of the use of
  the fresh relation symbols. Since $q''_3$ is a core by construction, and since the
  variables of the $q_i$-half and those of the $q_3$-half are distinguished by $A$ and
  $B$, we have
  \[\core(p_1)=\core(q'_1)\uplus\core(q'_3)
  \text{\quad and \quad}
  \core(p_2)=\core(q'_2)\uplus q''_3,
  \]
  and moreover every isomorphism between two such disjoint unions maps
  the $A$-marked half onto the $A$-marked half, and the $B$-marked half onto the $B$-marked
  half.

  Next, note that $\core(p_1)$ contains no atom whose relation symbol is
  one of the $B_x$, while $\core(p_2)$ contains all $n$ atoms of $\Delta$. Hence, for every
  bijective variable renaming $\rho$, the $n$ atoms of $\rho(\Delta)$ belong to
  $\core(p_1)\oplus\rho(\core(p_2))$. This shows that $\editdist(p_1,p_2)\geq n$ and,
  moreover, that
  \[\editdist(p_1,p_2)\leq n
  \quad\text{iff}\quad
  \core(p_1)\cong\core(p_2)\setminus\Delta,
  \]
  since an edit distance of $n$ means precisely that, up to renaming,
  $\core(p_1)$ is obtained from $\core(p_2)$ by deleting the atoms in $\Delta$ and nothing
  else.

  It remains to unfold the right-hand side. Deleting the atoms of $\Delta$
  from $q''_3$ gives back $q'_3$, and therefore
  $\core(p_2)\setminus\Delta=\core(q'_2)\uplus q'_3$. By the above,
  $\core(p_1)\cong\core(p_2)\setminus\Delta$ thus holds if and only if
  $\core(q'_1)\cong\core(q'_2)$ and $\core(q'_3)\cong q'_3$. The former is the case
  precisely when $q_1$ and $q_2$ are equivalent, and the latter precisely when $q'_3$ is a
  core, that is, when $q_3$ is a core.
  This finishes the proof of the claim.

  \smallskip
 
  For the $\Sigma^p_2$ upper bound, we use the following algorithm. Given input $q_1$, $q_2$ and $n$ with $q_1,q_2$ $k$-ary and over schema $\calS$, guess $k$-ary CQs $q'_1,q'_2$ over $\calS$ with no 
  more atoms than $q_1$ and $q_2$. Then verify  that $q'_i$ is the core of $q_i$, for $i \in \{1,2\}$, which can be done in DP \cite{FKPtods05}. We may then guess a variable renaming $\rho$
  and check in polynomial time that $|I_{q'_1} \oplus I_{\rho(q'_2)}| \leq n$. 
  
  \smallskip
  
  Now for the case where $q_1,q_2$ are promised to be cores. For the NP upper bound, the algorithm described above runs in NP because
 the guessing of
  cores and the DP check are unnecessary.
  For the lower bound, we exhibit a simple polynomial
  time reduction from the subgraph isomorphism problem, which is
  NP-complete, taking inspiration from \cite{DBLP:journals/pvldb/ZengTWFZ09}. 
  
  Let two undirected graphs $G_1=(V_1,E_1)$ and $G_2=(V_2,E_2)$ be given,
  for which we want to decide whether $G_1$ is isomorphic to a 
  subgraph of $G_2$. With a subgraph of $G_2$, we mean any graph $G=(V,E)$
  with $V \subseteq V_2$ and $E \subseteq E_2$. 
  Reserve two binary
  relation symbols $E$ and $R$. For $i \in \{1,2\}$, we construct a Boolean conjunctive query $q_i$ as follows:
  \begin{itemize}
      \item for all $(v_1,v_2) \in E_i$, include atoms $E(v_1,v_2)$ and $E(v_2,v_1)$;

      \item for all distinct $v_1,v_2 \in V_i$, include facts $R(v_1,v_2)$ and $R(v_2,v_1)$.
      
  \end{itemize}
  So $q_i$ is the ``overlay'' of $G_i$ with a clique. Because of the latter,
  $q_1$ and $q_2$ are cores. Moreover,  $G_1$ is isomorphic to a subgraph of  $G_2$ if and
  only if $\editdist(q_1,q_2) \leq n$ with $n = |e_{q_2}|-|e_{q_1}|$ (if this number is
  negative, then the answer is ``no'').

  For the only-if direction, assume that $G_1$ is isomorphic to a subgraph
  of $G_2$ via some isomorphism $\iota:V_1 \rightarrow V_2$, and let $\rho$ be a bijective
  variable renaming that extends $\iota^{-1}:\iota(V_1)\rightarrow V_1$, mapping the
  variables in $V_2\setminus\iota(V_1)$ to fresh variables. By the construction of $q_1$ and
  $q_2$, we have $e_{\iota(q_1)} \subseteq e_{q_2}$ and thus
  $e_{q_1}\subseteq e_{\rho(q_2)}$. Since $q_1$ and $q_2$ are cores, it follows that
  \[|\core(e_{q_1}) \oplus \rho(\core(e_{q_2}))| = |e_{\rho(q_2)} \setminus e_{q_1}|
  = |e_{\rho(q_2)}| - |e_{q_1}| = |e_{q_2}| - |e_{q_1}| = n.\]
 Thus $\rho$ witnesses
  $\editdist(q_1,q_2) \leq n$.

 Conversely, assume that $\editdist(q_1,q_2) \leq n$. Then there is a
  bijective variable renaming $\rho$ such that
  $|\core(e_{q_1}) \oplus \rho(\core(e_{q_2}))| \leq n$ and thus
  $|e_{q_1} \oplus e_{\rho(q_2)}| \leq n$. But by definition
  $n= |e_{q_2}| - |e_{q_1}| = |e_{\rho(q_2)}| - |e_{q_1}|$, implying
  $|e_{\rho(q_2)} \setminus e_{q_1}| \geq n$. It follows that
  $e_{q_1} \subseteq e_{\rho(q_2)}$, that is, $\rho^{-1}$ maps $e_{q_1}$ isomorphically onto
  a subinstance of $e_{q_2}$. By the construction of $q_1$ and $q_2$, the restriction of
  $\rho^{-1}$ to $V_1$ is then an isomorphism from $G_1$ onto a subgraph of $G_2$.
\end{proof}

We now move on to studying the proximity
pre-order $\preceq^{\editdist}$.
This pre-order
has some useful structural properties.
Let us say that a proximity pre-order $\preceq$
is \emph{well-founded} if for each
CQ $q$, every non-empty set of CQs 
has a $\preceq_q$-minimal element 
(i.e., there are no infinite  descending chains $\cdots \prec_q q_{2} \prec_q q_{1} \prec_q q_0$); and  $\preceq$
has the \emph{finite-basis property}
if for each
CQ~$q$, every set of CQs has only finitely many $\preceq_q$-minimal elements, up to equivalence.

\begin{prop}
\label{prop:editdist-finite-basis}
$\preceq^{\editdist}$ is well-founded and has the finite-basis property.
\end{prop}

\begin{proof}
Well-foundedness follows from the fact that the edit distance is a non-negative
integer and there are no infinite descending chains of non-negative integers. The finite basis property follows from the fact
that, for each non-negative integer~$n$, there are only
    finitely many CQs $q'$, up to equivalence, with $\editdist(q,q')\leq n$.
    This, in turn, is the case since there are, up to equivalence, only finitely many CQs  (over a given finite schema) of size at most $n$ for any given $n$.
\end{proof}

As an immediate consequence, we obtain:

\begin{thm}
\label{thm:editdist-finite-basis}
  Let $(q,E)$ be an annotated CQ.
  \begin{enumerate}
    \item If there is any CQ that fits $E$,  then there is a $\preceq^{\editdist}$-repair for $(q,E)$. 
  \item If there is any CQ $q'$ that fits $E$ with $q\subseteq q'$,  there is a  $\preceq^{\editdist}$-generalization for $(q,E)$. 
  \item
If there is any CQ $q'$ that fits $E$ with $q'\subseteq q$,  there is a $\preceq^{\editdist}$-specialization for $(q,E)$. \end{enumerate}
Moreover, there are at most finitely many
$\preceq^{\editdist}$-repairs, $\preceq^{\editdist}$-generalizations,
and $\preceq^{\editdist}$-specializations
for $(q,E)$, up to equivalence.
  \end{thm}

%
%
%
 %
%
%

%
%


We now compare $\preceq^{\editdist}$-repairs with $\preceq^{\dom}$-repairs.


\begin{exa}
  This example serves to compare 
  $\preceq^{\editdist}$-generalizations
  with $\preceq^{\dom}$-generalizations.
  Consider the following Boolean CQ and positive example:
  $$
  \begin{array}{rcl}
     q() &\colondash& R(x,y), R(x,z), P_1(y), P_2(y), Q_1(z), Q_2(z) \\[1mm]
     e &=& \!\!\{ R(a,b), R(a,c), P_1(b), Q_1(b), P_2(c), Q_2(c)\}
  \end{array}
  $$
  Let $E^+=\{e\}$. There are four $\preceq^{\editdist}$-generalizations for $(q,E^+)$,
  namely
  $$q'' \colondash
  R(x,y), R(x,z), P_i(y), Q_j(z) \text{ with } i,j\in\{1,2\}.
  $$
Indeed, $q$ itself does not fit $e$, because no element of $e$ satisfies
  both $P_1$ and $P_2$. For the same reason, no variable of a CQ that fits $e$ can carry
  both $P_1$ and $P_2$, and none can carry both $Q_1$ and $Q_2$. Since $q$ is a core, it
  follows that, for every CQ $q'$ that fits $e$ and every bijective variable renaming
  $\rho$, the set $e_q\oplus\rho(\core(e_{q'}))$ contains at least one of $P_1(y),P_2(y)$
  and at least one of $Q_1(z),Q_2(z)$; hence $\editdist(q,q')\geq 2$, with equality only if
  $\rho(\core(e_{q'}))$ is obtained from $e_q$ by deleting one atom of each of the two
  kinds. The four CQs $q''$ are precisely the CQs of this form, and each of them is a core
  that fits $e$ and satisfies $q\subseteq q''$.
  In contrast, there is a unique (up to equivalence) $\preceq^{\dom}$-generalization for $(q,E)$, namely
  $$
      q'() \colondash R(x,y), R(x,z), R(x,u), R(x,v), P_1(y), P_2(z), Q_1(u), Q_2(v).
  $$
  A variation of this example shows that (i) there can be exponentially more $\preceq^{\editdist}$-repairs than $\preceq^{\dom}$-repairs, and 
  (ii) $\preceq^{\dom}$-repairs can be exponentially larger than $\preceq^{\editdist}$-repairs.
  \end{exa}

  \begin{exa}
      Consider again Example~\ref{ex:dom-neg-notexists} which shows
  that  $\preceq^{\dom}$-specializations are not guaranteed to
  exist. There are multiple $\preceq^{\editdist}$-specializations,
  including the CQ
     $q'_1() \colondash R(x,x)$, the CQ
     $q'_2() \colondash R(x,y), R(y,z), R(z,x)$, and various different orientations of the latter triangle query. Each has edit distance 2.
  \end{exa}  

  \begin{exa}
  For the CQ $q(x)$ in Example~\ref{ex:dom-pos-inf-repairs}, where $\preceq^{\dom}$-repairs showed degenerative behavior, there exist  three $\preceq^{\editdist}$-repairs,
  namely $q'(x) \colondash R(x,y)$, $q''(x)\colondash R(y,x)$, and $q'''(x)\colondash R(x,x)$, each at edit distance~2.
\end{exa}


%

%

A 
  $\preceq^{\editdist}$-repair with respect to positive examples is not necessarily a 
  $\preceq^{\editdist}$-generalization:
  
\begin{exa}\label{ex:editdist-positive-not-a-generalization}
  Consider the following ``diamond-shaped'' Boolean CQs and example:
  $$
  \begin{array}{rcl}
      q() &\colondash& R(x,y_1), R(x,y_2), R(y_1,z), R(y_2,z), P(y_1), Q(y_2) 
  \\[1mm]
        q'() &\colondash& R(x,y_1), R(x,y_2), R(y_1,z), R(y_2,z), P(y_1), W(y_2) 
  \\[1mm]
      e &=& \{  R(a,b), R(a,c), R(b,d), R(c,d), P(b), W(c)\}
  \end{array}
  $$
  Let $E=(E^+,E^-)$ with $E^+=\{e\}$ and $E^-=\emptyset$.
  Then $q'$ is a $\preceq^{\editdist}$-repair for $(q,E)$, with edit distance~2,
  but it is not a $\preceq^{\editdist}$-generalization, since $W$ does not occur in $q$.
  To see that no fitting CQ is closer to $q$, note that every CQ that fits $E$ is $Q$-free, so
  that the atom $Q(y_2)$ always belongs to the symmetric difference, and that
  $\core(e_q)\setminus\{Q(y_2)\}$ is not a core (as $y_2$ can be mapped to $y_1$), so that no CQ has
  edit distance~1 to $q$ and fits $E$.
\end{exa}

Similarly, a 
  $\preceq^{\editdist}$-repair with respect to negative examples is not necessarily a 
  $\preceq^{\editdist}$-specialization:

\begin{exa}
\label{ex:editdist-negative-not-a-specialization}
  
   Consider the following Boolean CQs and negative example:
  $$
  \begin{array}{rcl}
      q() &\colondash& R(x,y_1), R(x,y_2), P(y_1), Q(y_1), W(y_2), P(y_2)
  \\[1mm]
        q'() &\colondash& R(x,y_1), R(x,y_2), P(y_1), Q(y_1), W(y_2), Q(y_2)
  \\[1mm]
      e &=&  \{  R(a,a), R(a,b), R(b,a), R(b,b), P(a), W(a), P(b), Q(b)\}
  \end{array}
  $$
Let $E=(E^+,E^-)$ with $E^+=\emptyset$ and $E^-=\{e\}$.
  A CQ maps homomorphically into $e$ if and only if none of its variables occurs both in a
  $Q$-atom and in a $W$-atom. Hence a CQ fits $E$ if and only if some variable occurs both in a $Q$-atom
  and in a $W$-atom. In particular, $q$ does not fit $E$, while $q'$ does, and
  $\editdist(q,q')=2$. No CQ that fits $E$ has edit distance~1 to~$q$: removing an atom from $q$
  yields a CQ that still maps into $e$, and the only atoms whose addition to $q$ creates a
  variable that occurs in both a $Q$-atom and a $W$-atom are $Q(y_2)$ and $W(y_1)$, but in both cases
  the resulting CQ is not a core ($y_1$ can be mapped to $y_2$ or vice versa). Thus $q'$ is a
  $\preceq^{\editdist}$-repair for $(q,E)$, but not a $\preceq^{\editdist}$-specialization,
  since $q'\not\subseteq q$.
  \end{exa}


The following upper bound on the size of
$\preceq^{\editdist}$-repairs is implied by the definitions.
\begin{prop}
\label{prop:size-of-editdist-repairs}
Let $(q,E)$ be an annotated CQ and 
$q'$ a core CQ.
If $q'$ is an $\preceq^{\editdist}$-repair for $(q,E)$,
then $|q'|\leq 2|q|+n$,\footnote{In the conference version of this paper, this bound was incorrectly stated as $|q|+n$.} where $n$ is the size of the smallest fitting CQ for $E$. The same holds for $\preceq^{\editdist}$-generalizations and for  $\preceq^{\editdist}$-specializations, where $n$ is then the
size of the smallest fitting CQ $q''$ that satisfies $q\subseteq q''$, respectively $q''\subseteq q$.
\end{prop}

\begin{proof}
Let $q_0$ be a fitting CQ for $E$ with $|q_0|=n$.
Since $q_0$ fits $E$ and $q'$ is a $\preceq^{\editdist}$-repair for $(q,E)$, we have
$\editdist(q,q')\leq\editdist(q,q_0)\leq |\core(e_q)|+n$. Moreover, since $q'$ is a core, we have
$|q'|=|\rho(\core(e_{q'}))|\leq |\core(e_q)|+|\core(e_q)\oplus\rho(\core(e_{q'}))| = |\core(e_q)|+\editdist(q,q')$,
where $\rho$ is a renaming that witnesses $\editdist(q,q')$. Together, this yields
$|q'|\leq 2|\core(e_q)|+n\leq 2|q|+n$, and the same argument applies to
$\preceq^{\editdist}$-generalizations and $\preceq^{\editdist}$-specializations, taking for
$q_0$ a smallest fitting CQ of the respective kind.
\end{proof}

Proposition~\ref{prop:size-of-editdist-repairs} stands in stark contrast with
Remark~\ref{remark:no-size-bound-dom-repairs} for $\preceq^{\dom}$-repairs. We
remark that, while the smallest fitting CQ may be exponential in the size of the input examples~\cite{Willard10},
one may expect it to be typically much smaller in practice.

We now consider algorithmic problems for $\preceq^{\editdist}$-repairs.

\begin{thm}
\label{thm:editdistboundedfittingcomplexity}
$\editdist$-\problem{bounded-fitting-existence} is $\Sigma^p_2$-complete if the distance bound  is given in unary. 
\end{thm}
\begin{proof}
For the $\Sigma^p_2$-upper bound we may use the following
algorithm. Assume that $(q,E)$ and $n$ are given as input and let the
number of atoms in $q$ be~$\ell$. Guess a CQ $q_0$ with at most $\ell$ atoms and a CQ $q'$with at most $\ell+n$ atoms, both over the same schema and of the same arity as $q$, and also guess a renaming $\rho$ of the variables in $q'$ that is the identity on all answer variables.
Then verify the following:
\begin{enumerate}
    \item $q_0$ is the core of $q$ and $q'$ is a core;
    \item $|I_{q_0} \oplus I_{\rho(q')}| \leq n$;
    \item $q'$ fits $E$.
\end{enumerate}
Point~1 can be verified in DP \cite{FKPtods05} and so can be Point~3 \cite{pods2023:extremal}, and Point~2 can obviously be verified in polynomial
time. We thus obtain the desired $\Sigma^p_2$-upper bound.

We prove the lower bound by reduction from {\sc size-bounded fitting existence} for
Boolean CQs, defined as follows:

\noindent
\textbf{Inputs}: 
\begin{itemize}
    \item a set $E$ of $0$-ary labeled examples; 
    \item a size bound $n \geq 0$;
\end{itemize} 
\textbf{Output}: 
\begin{itemize}
    \item  \emph{Yes} if there is a fitting CQ for $E$ with at most $n$ atoms and \emph{No} otherwise.
\end{itemize}

\noindent
This problem has been proved $\Sigma^p_2$-hard in \cite{GottlobLS99}. The 
reduction is as follows: $E,n$ form a yes-instance of  {\sc size-bounded fitting existence} if and only if $(q_\emptyset,E)$ and $d=n$ form a yes-instance of $\editdist$-\problem{bounded fitting existence} where $q_\emptyset$ is the empty Boolean CQ. The correctness of this reduction is easily verified. 
\end{proof}

By Theorem~\ref{thm:editdist-finite-basis}, 
the $\preceq^{\editdist}$-\problem{repair-existence} problem coincides with the fitting existence problem, in the following sense.
Let $\calS$ be the set of relation symbols occurring in $(q,E)$, c.f.\ Remark~\ref{rem:no-other-relations}.
By Theorem~\ref{thm:editdist-finite-basis}, a
$\preceq^{\editdist}$-repair for $(q,E)$ exists if and only if
there is a CQ over~$\calS$, of the same arity as $q$, that fits~$E$.
Note that, consequently, repair
existence does not depend on $q$ itself, but only on the relation symbols used in it. The verification and construction problems for
$\preceq^{\editdist}$-repairs are more interesting and depend on the particular input query.
The same is true for $\preceq^{\editdist}$-generalizations
and $\preceq^{\editdist}$-specializations.
\begin{thm}
\label{thm:editdistrepaircomplexity}
~
    \begin{enumerate}
    \item $\preceq^{\editdist}$-\problem{repair-existence}
    is coNExpTime-complete.
    \item $\preceq^{\editdist}$-\problem{repair-existence} is coNP-complete
    for inputs that consist of a bounded number of examples.
    \item $\preceq^{\editdist}$-\problem{repair-existence} is in PTime
    if the input contains only positive examples, or only negative examples.
   
    \item $\preceq^{\editdist}$-\problem{repair-verification}  is  $\Pi^p_2$-complete.


    \end{enumerate}
     Items~1 and~4 also hold if ``repair'' is replaced by ``generalization'' or
    ``specialization''. Item~2 also holds for generalizations, while for specializations we only
    have a DP upper bound. Item~3 also holds for generalizations if the input contains only
    positive examples, and for specializations if the input contains only negative examples.\footnote{In the conference version of this paper, it was incorrectly claimed that the PTime upper bound carries over to generalizations and specializations without these restrictions, and the DP upper bound for a bounded number of examples was incorrectly stated for generalization-existence instead of specialization-existence.}
\end{thm}

\begin{proof}
\emph{Points~1--3}. For $\calS$ the set of relation symbols occurring in $(q,E)$, a $\preceq^{\editdist}$-repair
for  $(q,E)$ exists if and only if there is a CQ $q'$ that fits
$E$: we can always edit $\core(q)$ to arrive at $\core(q')$. Thus
 $\editdist$-{\sc repair existence} coincides with existence of a fitting CQ for $E$, which is coNExpTime-complete with an unbounded number of examples and coNP-complete when the number of examples is bounded by a constant \cite{CateD15}.
 If the input contains only positive examples or only negative examples,
 existence of a fitting CQ can be decided in PTime. We only sketch the algorithms, starting
 with the case of only positive examples. Here, the only reason why a fitting may not exist is that in the
 product $e_\Pi$ of the positive examples, some distinguished element does not participate in
 any fact; as a consequence, the CQ $q_\Pi$ whose canonical example is $e_\Pi$ is then not safe,
 and in fact there is no (safe) fitting CQ. We may thus decide fitting existence for a given set of $k$-ary labeled examples $E$ by checking that, for every $i \in \{1,\dots,k\}$, there is a
 relation symbol $R$ of some arity $n$ and a $j \in \{1,\dots,n\}$  such that for all
 positive examples $(I,\textbf{a}) \in E^+$, the $i$-th value $a_i$ in {\bf a} participates in some
 fact $R(\bar c) \in I$ in the $j$-th position.  In the case of only negative examples, let
 $q_\bot^{k}$ be the CQ with answer variables $x_1,\dots,x_k$ and a single existential variable $z$,
 whose body consists of all atoms $R(u_1,\dots,u_m)$ where $R$ is an $m$-ary relation symbol
 occurring in $\calS$ and $u_1,\dots,u_m \in \{x_1,\dots,x_k,z\}$. Then $q_\bot^{k} \subseteq q'$ for every
 CQ $q'$ over $\calS$. Since
 fitting a set of negative examples is preserved under passing to a more specific CQ, it follows
 that if some CQ fits $E$, then $q_\bot^{k}$ fits $E$, and hence a fitting CQ exists if and only if
 $q_\bot^{k}$ fits $E$. The latter can be tested in polynomial time, even without assuming a bound
 on the arity of the relations. Indeed,  $q_\bot^{k}$ fails to fit
 $(I,\mathbf{a})$ precisely if some $c \in \adom(I)$ is such that $I$ contains every fact over
 $D_c = \{a_1,\dots,a_k,c\}$ that uses a relation symbol occurring in~$\calS$. There are at most
 $|\adom(I)|$ candidates $c$, and for a fixed $c$ the condition can be tested in polynomial time
 (by counting, for each relation symbol $R$, the $R$-facts with all values in $D_c$ and  comparing
 this count with $|D_c|^{arity(R)}$).

 The above arguments are easily varied to apply to generalization and specialization as a repair, we now give details.
First, there is a $\preceq^{\editdist}$-generalization for $(q,E)$  if and only if there is a CQ $q'$ that fits the set of examples $E'$ that is obtained from $E$ by adding the canonical  example $e_q$ of $q$ as an additional positive example. 
And second, there is a $\preceq^{\editdist}$-specialization
for $(q,E)$ if and only if
(i)~there is a CQ over $\calS$ that fits $E$, and
(ii)~$q$ fits $E^+$.
Necessity follows because every specialization is contained
in $q$. For sufficiency, let $r$ be a CQ over $\calS$ that
fits $E$. If $q$ fits $E^+$, then $q\land r$ fits $E$ and
is contained in $q$. The existence of a
$\preceq^{\editdist}$-specialization now follows from
Theorem~\ref{thm:editdist-finite-basis}.

And second, there is a $\preceq^{\editdist}$-specialization for $(q,E)$  if and only if (i)~there is a CQ that fits 
$E$ and (ii)~$e_q \rightarrow \Pi_{e\in E^+} e$. To see this, note that the
CQ $q_\Pi$ whose canonical example is 
$\Pi_{e\in E^+} e$ is the most 
specific CQ that fits $E$ \cite{pods2023:extremal}. Thus any 
specialization $q'$ of $q$ that fits $E$ must satisfy $q_\Pi \subseteq q'$.

Moreover, if (i) and (ii) are satisfied, then since  $\preceq^{\editdist}$ is well-founded we can start
from the fitting CQ $q_\Pi$ and by successive generalization find a $\preceq^{\editdist}$-specialization for $(q,E)$. It remains to remark that testing condition (i) is coNExpTime-complete and testing condition~(ii) is in NP, so that the overall problem is in coNExpTime. For a bounded number of examples, the same approach yields a DP algorithm. The PTime cases of Item~3 follow as well: if $E$ contains only positive examples, then so does $E'$, and if $E$ contains only negative examples, then condition~(ii) is vacuous and condition~(i) is decidable in PTime, as shown above.

\emph{Point~4}. We start with the upper bound. Assume that $(q_1,E)$ and $q_2$ are given as the input,
and let $m$ be the number of atoms in $q_1$ plus the number of atoms in $q_2$.
We show that the complement of the problem is in $\Sigma^p_2$. Clearly, $q_2$
fails to be a $\preceq^{\editdist}$-repair for $(q_1,E)$ if and only if $q_2$ does not fit
$E$, or there is a CQ $q'$ that fits $E$ with $\editdist(q_1,q')<\editdist(q_1,q_2)$. The former
can be tested in DP. For the latter, we may assume that $q'$ is a core, and then
$|q'|\leq|\core(e_{q_1})|+\editdist(q_1,q')\leq 2m$, by the argument in the proof of
Proposition~\ref{prop:size-of-editdist-repairs}. Thus, we guess a CQ $q'$ with at most $2m$
atoms, examples $C_1\subseteq e_{q_1}$ and $C_2\subseteq e_{q_2}$, and a
bijective variable renaming $\rho$ that is the identity on the answer variables, and we verify that
\begin{enumerate}
    \item $C_1=\core(e_{q_1})$, $C_2=\core(e_{q_2})$ and $q'$ is a core;
    \item  $q'$ fits $E$;
    \item  $|C_1\oplus e_{\rho(q')}|<|C_1\oplus\rho'(C_2)|$ for every bijective variable renaming $\rho'$
    that is the identity on the answer variables.
\end{enumerate}
Point~1 and~2 are in DP, and Point~3 is in coNP. Since,
by Point~1, $C_1$ and $C_2$ are the cores of $e_{q_1}$ and $e_{q_2}$ and $e_{q'}$ is its own core,
Point~3 expresses that $\editdist(q_1,q')<\editdist(q_1,q_2)$.
Altogether, this is a nondeterministic polynomial-time algorithm with access to an NP-oracle,
which shows that the complement of $\preceq^{\editdist}$-\problem{repair-verification} is in
$\Sigma^p_2$, and hence the problem itself is in $\Pi^p_2$.

The lower bound is proved by reduction from the complement of \editdist-{\sc bounded fitting existence}, which is $\Pi^p_2$-hard by Theorem~\ref{thm:editdistboundedfittingcomplexity}. The proof  shows that this problem 
is hard already for Boolean CQs, and in fact already for the fixed Boolean empty-body
CQ $q_\emptyset$. We may thus restrict our attention to this case. Let $(q,E)$ and $d$ be an input to this restricted version, with $q=q_\emptyset$. Reserve a fresh binary relation symbol $R$ and let $E'$ be obtained from $E$ by disjointly adding to each positive example an $R$-path of length $d+1$ and to each negative example an $R$-path of length $d$. 

Let $q_R$ be the Boolean CQ whose canonical example is an $R$-path of length $d+1$.
We claim that
the following are equivalent:
\begin{itemize}
    \item $(q_\emptyset,E)$ and $d$ is a no-instance of  \editdist-{\sc bounded fitting existence};
    \item $(q_\emptyset,E')$ and $q_R$ is a yes-instance of $\preceq^{\editdist}$-{\sc repair verification}.
\end{itemize}
To see this, note first that $\core(e_{q_\emptyset})=\emptyset$ and thus
$\editdist(q_\emptyset,q')=|\core(q')|$ for every CQ $q'$; the first item therefore says
that no CQ that fits $E$ has a core with at most $d$ atoms. Moreover, $q_R$ fits $E'$,
because the $R$-path of length $d+1$ maps homomorphically into each positive example of
$E'$ but not into any negative example, which contains only an $R$-path of length $d$; and
$\editdist(q_\emptyset,q_R)=d+1$, because $e_{q_R}$ is a core. Hence the second item says
that no CQ that fits $E'$ has a core with at most $d$ atoms. It remains to observe that
such a CQ exists for $E'$ if and only if it does for $E$. For the one direction, a CQ that
fits $E$ contains no $R$-atoms and hence also fits $E'$. For the other direction, let $q''$
fit $E'$ with $|\core(q'')|\leq d$. Since $R$ occurs only in the added paths and these
contain no other facts, every connected component of $q''$ that contains an $R$-atom
consists of $R$-atoms only and maps into the $R$-path of length $d+1$; having at most $d$
atoms, such a component even maps into the $R$-path of length $d$ and thus into every
negative example of $E'$. The components that make $q''$ fail on the negative examples are
therefore $R$-free, and dropping all remaining components yields a CQ that fits $E$ and
whose core has at most $d$ atoms.

\smallskip
When generalizations or specializations are considered in place of repairs, 
essentially the same upper bound proof applies.  We only need to verify, 
in addition, that $q_1 \subseteq q_2$ or $q_2 \subseteq q_1$, respectively, and to restrict the
guessed CQ $q'$ accordingly; since CQ containment is in NP, this does not change the complexity of the overall
procedure. The lower bound proof above clearly also applies to specializations,
but not to generalizations. 

To prove $\Pi^p_2$-hardness for the latter, 
we first observe that in the hardness proof for {\sc size-bounded fitting existence} given
in  \cite{GottlobLS99}, the (only) positive example does not contain any reflexive
facts, that is, no facts of the form $R(c,\dots,c)$. This implies that
$\editdist$-\problem{bounded-fitting-existence} is $\Sigma^p_2$-hard already in
the restricted case where  positive examples contain no reflexive facts and for the fixed Boolean CQ $q_{\max}$ that contains only a single 
variable $x$ and a reflexive atom $S(x,\dots,x)$ for every relation
symbol $S$ used in the hardness proof in~\cite{GottlobLS99} as well as for
a fresh binary relation symbol $R$. Let $\ell$ be
the number of such symbols. In fact, $\Sigma^p_2$-hardness for the restricted
version of $\editdist$-\problem{bounded-fitting-existence} just described again follows  from a simple reduction from {\sc size-bounded fitting existence}: if the positive examples in $E$
contain no reflexive facts, then $E,d$ form a yes-instance of  {\sc size-bounded
fitting existence} if and only if  there is a CQ that fits $E$ and has edit distance at most $\ell+d$ to $q_{\max}$.
This is because we have to remove all $\ell$ reflexive atoms from $q_{\max}$
to obtain a fitting, then arriving at $q_\emptyset$. 

We can now prove $\Pi^p_2$-hardness of $\preceq^{\editdist}$-\problem{generalization-verification}
by reduction from the complement of the restricted version of  $\editdist$-\problem{bounded-fitting-existence}. Let $(q,E)$ and $d$ be an input for this problem, with $q=q_{\max}$. We can assume that $d \geq \ell$: since the positive examples in $E$
contain no reflexive facts and we are using the fixed query $q_{\max}$, it is 
otherwise clear that the answer is `no'. With $R$ the fresh binary relation symbol fixed above, let $E'$ be obtained from $E$ by disjointly adding to each positive example an $R$-path of length $(d-\ell)+1$ and to each negative example an $R$-path of length $d-\ell$. Let $q_R$ be the Boolean CQ whose canonical example is an $R$-path of length $(d-\ell)+1$. The following are equivalent:
\begin{itemize}
    \item $(q_{\max},E)$ and $d$ is a no-instance of  \editdist-{\sc bounded fitting existence};
    \item $(q_{\max},E')$ and $q_R$ is a yes-instance of $\preceq^{\editdist}$-\problem{generalization-verification}.
\end{itemize}
Indeed, since $q_{\max}$ contains a reflexive atom for every relation symbol of the
schema, every CQ $q'$ satisfies $e_{q'} \rightarrow e_{q_{\max}}$ and thus $q_{\max} \subseteq q'$;
hence $\preceq^{\editdist}$-generalizations for $(q_{\max},E')$ are precisely the
$\preceq^{\editdist}$-repairs for $(q_{\max},E')$, and $q_R$ in particular is a candidate.
Moreover, the positive example in $E'$ contains no reflexive fact, so $\core(e_{q''})$ contains no
reflexive atom for any CQ $q''$ that fits $E'$, and therefore
$\editdist(q_{\max},q'')=\ell+|\core(q'')|$; in
particular $\editdist(q_{\max},q_R)=\ell+((d-\ell)+1)=d+1$. The equivalence now follows exactly as
in the case of repairs above: both items express that no CQ fitting $E$ has a core with at most
$d-\ell$ atoms.
We have thus found the desired reduction.
\end{proof}

Finally, let us discuss $\preceq^{\editdist}$-\problem{repair-construction}. 
It follows from known results~\cite{CateD15} that 
whenever a fitting CQ exists and there is at least one positive example, 
then there is one of size at most $n_{max}=\Pi_{e\in E^+}|e|$. A brute-force algorithm for $\preceq^{\editdist}$-\problem{repair-construction}  simply
enumerates CQs in the order of increasing size and, for each, checks if it is a $\preceq^{\editdist}$-repair (cf.~Theorem~\ref{thm:editdistrepaircomplexity}(4)). 
Since we are promised that an $\preceq^{\editdist}$-repair exists, 
 this process
terminates and, by Proposition~\ref{prop:size-of-editdist-repairs}, yields a CQ of size at most $2|q|+n_{max}$.
We do not know of an algorithm for constructing $\preceq^{\editdist}$-repairs with
asymptotically better running time,
see Section~\ref{sec:discussion} for further discussion. In the special
case that there are no positive examples, constructing a fitting CQ is not very interesting. If a fitting CQ exists and the examples have arity $k$, then also the CQ fits that has answer variables $x_1,\dots,x_k$, one existential variable $y$, and all possible atoms over these variables. Note that the size of this CQ only depends on the arity and schema, but not on the 
concrete examples.



\begin{rem} 
\label{rem:finer-edit-distance}
One can further refine edit distance by requiring that
all equalities between variables in the
body of the query are represented explicitly by means of equality atoms, and by  counting these when computing the symmetric difference.
For example, consider the CQs
\[\begin{array}{lll}
q(x) &\colondash& P(x)   \\[1mm]
 q'_1(x) &\colondash& P(x), R(y,z)   \\[1mm]
 q'_2(x) &\colondash& P(x), R(y,y)
\end{array}\]
Under Definition~\ref{def:editdist}, 
$\editdist(q,q'_1)= \editdist(q,q'_2)$,
while under the more refined definition of edit distance
(where we treat $q'_2$  as shorthand for $q'_2(x) \colondash P(x), R(y,z), y=z$), 
$\editdist(q,q'_1)< \editdist(q,q'_2)$. 
We omit the details, but we believe that the above results continue to hold under such a more refined notion of edit distance.
\end{rem}


\subsection{Other Distance Metrics}

We review two other approaches to defining a distance metric between queries.

\paragraph{Distance as size of the smallest distinguishing instance}
The next distance metric we consider
is based on the smallest instance
on which the two queries produce different answers.

\begin{defi}
\label{def:sdidist}
    For CQs $q_1$ and $q_2$,
    $\sdidist(q_1, q_2)=1/(n+1)$, where $n$ is the size of the smallest instance $I$
    (measured by the number of facts) such that
    $q_1(I)\neq q_2(I)$, or 0 if no such instance $I$ exists.
    \footnote{We divide by $n+1$ instead of $n$ to account for the corner case where
    the smallest distinguishing instance is empty, which can happen  when the two
    queries are Boolean and one of the two has an empty body. In the conference version of this paper, $\sdidist$ was defined as $1/n$, which is undefined in this corner case; Propositions~\ref{prop:sdidist-bound} and~\ref{prop:sdidist-complexity} below have been adjusted accordingly.}
\end{defi}

\begin{prop}
\label{prop:sdidist-is-a-metric}
$\sdidist$ is a semantic distance metric, and in fact an ultrametric (i.e., it satisfies 
$dist(q_1,q_3)\leq \max(dist(q_1,q_2), dist(q_2,q_3))$).
\end{prop}

\begin{proof}
Recall that an ultrametric is a distance metric satisfying
$dist(q_1,q_3)\leq \max(dist(q_1,q_2), dist(q_2,q_3))$.
It is clear 
    that $dist(q,q')$ is a non-negative real number; that $dist(q,q')=0$ if and only
    if $q$ and $q'$ are equivalent; 
    and that $dist(q_1,q_2)=dist(q_2,q_1)$. It remains to observe that
    $\sdidist(q_1,q_3)\leq \max(\sdidist(q_1,q_2), \sdidist(q_2,q_3))$.
     Let $n$, respectively $m$, be the size of the smallest instance on
    which $q_1$ and $q_2$, respectively $q_2$ and $q_3$, disagree (where
    $n,m\in\mathbb{N}\cup\{\infty\}$), so that $\sdidist(q_1,q_2)=1/(n+1)$ and
    $\sdidist(q_2,q_3)=1/(m+1)$. Let $I$ be any instance containing fewer than
    $\min(n,m)$ facts. Then $q_1(I)=q_2(I)$ and $q_2(I)=q_3(I)$, hence $q_1(I)=q_3(I)$. It
    follows that every instance on which $q_1$ and $q_3$ disagree has at least
    $\min(n,m)$ facts, and therefore
    $\sdidist(q_1,q_3)\leq 1/(\min(n,m)+1) = \max(1/(n+1),1/(m+1))$.
\end{proof}

\begin{exa}
Consider the CQs 
\[q_1() \colondash R(x,x)\] and 
\[q_2() \colondash \bigwedge_{i, j\in\{1, \ldots, N\}, i\neq j} R(x_i,x_j).\]
An instance that satisfies $q_2$ either isomorphically embeds a clique of size $N$ or else contains a ``reflexive'' fact of the form $R(a,a)$. Therefore, 
every example distinguishing $q_1$ from $q_2$
must contain at least $N(N-1)$ facts.
It follows that $\sdidist(q_1,q_2)=1/(N(N-1)+1)$.
\end{exa}


Two further relevant basic facts about $\sdidist$ are the following:

\begin{prop}
    \label{prop:sdidist-bound}
    For all CQs $q,q'$, either $\sdidist(q,q')= 0$ or $\sdidist(q,q')\geq 1/(\max(|q|,|q'|)+1)$.
\end{prop}

\begin{proof}
    If two CQs $q,q'$ are non-equivalent, they must disagree either
    on the canonical instance of $q$ or on the canonical instance of $q'$.
    Hence the smallest instance on which they disagree has at most
    $\max(|q|,|q'|)$ facts.
\end{proof}

\begin{prop}
  \label{prop:sdidist-complexity}
  Computing $\sdidist$ is NP-hard. More precisely,
testing $\sdidist(q,q')\leq 1/k$ (on input CQs $q,q'$ and natural number $k\geq 1$ in unary) is NP-hard and is in $\Pi^p_2$.
\end{prop}

\begin{proof}
For $k=1$, the answer is trivially \emph{yes}, because $\sdidist$ takes
values in the interval $[0,1]$. For $k\geq 2$, we have $\sdidist(q,q')\leq 1/k$ if and
only if the smallest instance on which $q$ and $q'$ disagree has at least $k-1$ facts,
that is, if and only if $q$ and $q'$ agree on all instances with fewer than $k-1$ facts.
The problem is therefore in $\Pi^p_2$ by co-guessing an instance with fewer than $k-1$
facts and testing that $q$ and $q'$ disagree on it. NP-hardness holds because for
$k>\max(|q|,|q'|)+1$ we have, by Proposition~\ref{prop:sdidist-bound}, that
$\sdidist(q,q')\leq 1/k$ if and only if $\sdidist(q,q')=0$, that is, if and only if $q$
and $q'$ are equivalent.
\end{proof}

We will now move on to study the pre-order $\preceq^\sdidist$. 
The next example shows that $\preceq^\sdidist$ is not  a 
good pre-order for query repairs, 
since it is not sufficiently discriminative.

\begin{exa}\label{ex:sdidist-trivializes}
    Let \[q_1(x) = R(x,y),\] and let $E_1$ consist of 
    \begin{itemize}
        \item the negative example $(I,a)$ where $I=\{R(a,b)\}$, and
        \item the positive example $(J,a)$ where $J=\{ R(a,a)\}$.
    \end{itemize}
    The positive example is strictly speaking redundant: every CQ over the relevant schema fits it. It is added only for intuition.   
    There are infinitely many pairwise non-equivalent CQs $q'$ that fit $E_1$ with $q'\subseteq q_1$. Furthermore, 
    by construction, every fitting CQ $q'$ disagrees with  $q_1$ on $I$, and hence has  $\sdidist(q',q_1)=1/2$. It follows that all infinitely-many fitting CQs are 
    $\preceq^{\sdidist}$-repairs for $(q_1,E_1)$, and 
    there are infinitely many
    $\preceq^{\sdidist}$-specializations for $(q_1,E_1)$ as well.

    A similar situation arises for $\preceq^{\sdidist}$-generalizations: let \[q_2(x) = R(x,x,x)\] and 
    let $E_2$ consist of the positive example $(I',a)$ where
    $I'=\{R(a,a,b)\}$. There are infinitely many CQs  (up to equivalence)  that fit $E_2$ and they all 
    disagree with $q_2$ on the single-fact instance $I'$.
    For instance, for each $n\geq 1$, one may take
    \[q'_n(x) \colondash R(x,v_1,w), R(v_1,v_2,w), \ldots, R(v_{n-1},v_n,w).\]
    \end{exa}

This shows that there are annotated CQs with infinitely many $\preceq^{\sdidist}$-repairs (as well as
$\preceq^{\sdidist}$-specializations and 
$\preceq^{\sdidist}$-generalizations). 
On the flip side, we have:

\begin{prop}
    Let $(q,E)$ be an annotated CQ.
    \begin{enumerate}
        \item If there is any CQ that fits $E$, then 
    there is a $\preceq^{\sdidist}$-repair for     $(q,E)$.
    \item If there is any CQ $q'$ that fits $E$  with $q\subseteq q'$,  
    there is a $\preceq^{\sdidist}$-generalization for     $(q,E)$.\item If there is any CQ $q'$ that fits $E$  with $q'\subseteq q$,  
    there is a $\preceq^{\sdidist}$-specialization for     $(q,E)$.
    \end{enumerate}
\end{prop}

\begin{proof}
    We prove the first item. The proof of the other items is similar.
    If $q$ fits $E$,  $q$ is its own 
    $\preceq^{\sdidist}$-repair. Therefore,
    assume that $q$ does not fit $E$, and 
    let $n$ be the size of the smallest 
    example $e$ in $E$ that $q$ does not fit.
    Every CQ $q'$ that fits $E$
    must disagree with $q$ on $e$, and
    therefore satisfies $\sdidist(q,q')\geq 1/(n+1)$.
    In other words, every fitting CQ $q'$ must
    satisfy $\sdidist(q,q')=1/(m+1)$ for some
    $m\in\{0,1, \ldots, n\}$. Since this range of possible distance values is finite and a fitting CQ
    exists, it follows that a $\preceq^{\sdidist}$-repair exists.
\end{proof}

We omit a complexity-theoretic analysis of $\preceq^\sdidist$-repairs in light of Example~\ref{ex:sdidist-trivializes}.


\begin{rem}
Dually to $\sdidist$ one can also consider
the distance metric $\sdqdist$ based on
the size of the smallest distinguishing query.
More precisely,  $\sdqdist(q,q')=1/n$, where $n$ is
the size of the smallest CQ (as measured by the number of atoms) that maps homomorphically to 
precisely one of $q,q'$, or 0 if no such CQ exists.
\footnote{Here, no corner case as in Definition~\ref{def:sdidist} arises: the
empty-body Boolean CQ maps homomorphically to every instance and hence never
distinguishes two queries, so that $n\geq 1$.}
Unfortunately, $\sdqdist$ fares no better than $\sdidist$. For instance, for the annotated CQ $(q_2,E_2)$ from Example~\ref{ex:sdidist-trivializes}, all CQs $q'$ that fit $E_2$ again have $\sdqdist(q',q_2)=1$.
\end{rem}




\paragraph{Distance as probability of disagreement}
Let $\mu$ be a discrete probability distribution over the space of all examples (an \emph{example distribution}, for short). For instance, $\mu$ may be a uniform distribution over values in some pre-existing (unlabeled) database instance. 
We define $dist_\mu(q,q')=\mu(\extension{q}\oplus \extension{q'})$. That is, $dist_\mu(q,q')$
is the probability of drawing an example on which $q$ and $q'$ disagree.
The same works for non-discrete probability distributions, as long as $\extension{q}$ is measurable for each CQ $q$. We  restrict to discrete  distributions for simplicity.
It is easy to see that, for all example distributions $\mu$,
     $dist_\mu$ is a weak semantic distance metric.

%

\begin{prop}
    \label{prop:computing-mu-dist}
    $dist_\mu(q,q')$ can be computed in polynomial time using a single
    parallel round of calls to an NP-oracle
    (for example distributions $\mu$ with finite support, specified as part of the input).
    Testing $dist_\mu(q,q')\leq r$  is  $P^{NP}_{||}$-complete.
\end{prop}

\begin{proof}
We recall that the complexity class $P^{NP}_{||}$ admits 
many equivalent definitions, including
as the class of problems solvable in polynomial time using a single parallel round of calls to an NP-oracle (meaning that the algorithm may make polynomially many calls to an NP-oracle but these calls must be independent of each other)~\cite{Buss1991:truth}.
    Following this definition, we use  NP-oracle calls to
    test the label of each example in the support of $\mu$ with respect to $q$ and with respect to $q'$. We add up the probabilities of the examples where the labels according to $q$ and $q'$ differ. This way we compute  
    $dist_\mu(q,q')$. We then compare with $r$. 
    For the lower bound,
    we give a reduction from the problem of testing whether two graphs have the same chromatic number. It follows
    from results in~\cite{Wagner1986:more} and \cite{Buss1991:truth} that this problem is complete for $P^{NP}_{||}$. Take $q_1=q_{G_1}$ and $q_2=q_{G_2}$ to be the canonical CQs of the two given graphs, and take $\mu$ to be the uniform distribution over all cliques of size at most the size of $max(size(G_1),size(G_2))$. Then $dist(q_1,q_2)=0$ iff $G_1$ and $G_2$ have the same chromatic number.
\end{proof}

For the sake of readability, we will denote the proximity pre-order
$\preceq^{dist_\mu}$ by $\preceq^\mu$. 

\begin{prop}
    Let $\mu$ be any example distribution.
    If $\mu$ has finite support, $\preceq^{\mu}$ is well-founded. The same does not necessarily hold when  $\mu$ has infinite support.
    \label{prop:mu-well-founded}
\end{prop}

\begin{proof}
    If $\mu$ has finite support, say of cardinality $n$, then there are only 
    $2^n$ possible values for the probability of an event.
    It follows that there cannot exist 
    an infinite sequence of events of strictly decreasing probability. 
    It follows that $\preceq^\mu$ is well founded.

    For the second part of the proposition, consider the schema
    $\calS=\{R\}$ with $R$ binary. For $j\geq 1$, let $A_j$ be the instance that consists
    of a directed $R$-path with $j$ edges, and let $\mu(A_j)=2^{-j}$. Let $q$ be the
    Boolean CQ \[q() \colondash R(x,y),R(y,z),R(z,x)\] and, for $k\geq 1$, let $q_k$ be the
    Boolean CQ whose canonical instance is a directed $R$-path with $k$ edges. Then $q$ is
    false on every $A_j$, because $A_j$ is acyclic, whereas $q_k$ is true on $A_j$ if and
    only if $k\leq j$. Therefore $q$ and $q_k$ disagree precisely on the instances $A_j$
    with $j\geq k$, and hence $dist_\mu(q,q_k)=\Sigma_{j\geq k}2^{-j}=2^{-(k-1)}$. Since
    these values are strictly decreasing, $q_{k+1}\prec^{\mu}_q q_k$ holds for all
    $k\geq 1$, and thus $\preceq^\mu$ is not well-founded.
\end{proof}

The failure of well-foundedness does not by itself imply that some
annotated CQ lacks a $\preceq^\mu$-repair, because a repair only needs to be minimal among
the CQs that \emph{fit} the given examples. The next proposition shows that, for
distributions with infinite support, repairs may indeed fail to exist.

\begin{prop}
\label{prop:mu-infinite-no-repair}
There are an annotated CQ $(q,E)$ and an example distribution $\mu$ such that fitting CQs for $E$ exist but there is no
$\preceq^{\mu}$-repair for $(q,E)$.
\end{prop}

\begin{proof}
     Let us fix a schema $\calS=\{R,P\}$, where $R$ is a binary relation (think: the edge
    relation of a directed graph) and $P$ is a unary relation. We construct $\mu$ by
    giving probability $1/2$ to the instance $B$ that consists of the two facts
    $R(b,c),R(c,b)$ and no $P$-facts, and by dividing the remaining probability mass $1/2$
    over all non-empty instances $A$ satisfying $P^A=\adom(A)$, in such a way that the
    probabilities sum up to 1. For instance, this can be done by choosing a suitable
    enumeration $A_1, A_2, \ldots$ of these instances (up to isomorphism) and assigning
    probability $\mu(A_i)=2^{-(i+1)}$.

    Let $E$ be the collection of labeled examples that consists
    of a single positive example, namely the one-element $R$-loop $\{R(a,a)\}$ without
    $P$-facts, and let $q$ be the Boolean CQ $q() \colondash P(x)$. By
    construction, $q$ fails to fit $E$. Moreover, $q$ is false on $B$, while $q$
    is true on every other instance to which $\mu$ assigns a positive probability. Note
    also that a CQ $q'$ fits $E$ if and only if $q'$ contains no occurrence of
    $P$: indeed, the positive example contains no $P$-facts, while every CQ without
    occurrences of $P$ maps homomorphically into it (map all variables to $a$).

    Suppose, for the sake of a contradiction, that $q'$ is a
    $\preceq^\mu$-repair for $(q,E)$. By the above, $q'$ contains no occurrence
    of $P$ and hence $q'$ is the canonical CQ $q_G$ of a directed graph $G$. We first
    observe that $q'$ must be false on $B$. Indeed, otherwise $q$ and $q'$ disagree on
    $B$, and hence $dist_\mu(q,q')\geq 1/2$, whereas the CQ $q''$ expressing the existence
    of a directed $R$-cycle of length 3 fits $E$ and satisfies
    $dist_\mu(q,q'')<1/2$: the queries $q$ and $q''$ agree on $B$ (both are false on $B$,
    the latter because $B$ is 2-colorable while a cycle of odd length is not) and they
    agree on at least one instance in the support of $\mu$ of positive probability,
    namely on the directed $R$-cycle of length 3 in which all elements belong to $P$.

    So $q_G$ is false on $B$, which means that $G\not\to B$ and hence
    that $G$ is not 2-colorable; consequently, $G$ contains an undirected cycle of odd
    length, say of smallest length $n$. By the \emph{sparse incomparability
    lemma}~\cite{kun2013constraints}, there is a directed graph $G'$ of girth strictly
    greater than $n$ such that $G'\to G$ and such that $G$ and $G'$ map homomorphically to
    precisely the same instances of size at most $|B|$; in particular, $G'\not\to B$.
    Then $q_{G'}$ fits $E$, it is false on $B$ (just like $q$ and $q_G$), and
    $q_G\subseteq q_{G'}$, so that $\extension{q_G}\subseteq\extension{q_{G'}}$. Since $q$
    is true on every instance in the support of $\mu$ except for $B$, it follows that the
    set of instances in the support of $\mu$ on which $q$ and $q_{G'}$ disagree is a
    subset of the set of instances on which $q$ and $q_G$ disagree. Moreover, the
    inclusion is strict: the instance $A$ that extends $G'$ with the fact $P(v)$ for every
    element $v$ of $G'$ belongs to the support of $\mu$, and $q_{G'}$ is true on $A$,
    while $q_G$ is false on $A$, because $G\to G'$ would imply that $G'$ contains an
    undirected cycle of odd length at most $n$, contradicting the fact that the girth of
    $G'$ is greater than $n$. Therefore $dist_\mu(q,q_{G'})<dist_\mu(q,q_G)$,
    contradicting the assumption that $q'=q_G$ is a $\preceq^\mu$-repair for
    $(q,E)$.
\end{proof}

As a consequence of Proposition~\ref{prop:mu-well-founded}, we obtain:

\begin{prop}
    Let $(q,E)$ be an annotated CQ and $\mu$  an example distribution with finite support.
  \begin{enumerate}
    \item If there is any CQ that fits $E$,  then there is a $\preceq^{\mu}$-repair for $(q,E)$. 
  \item If there is any CQ $q'$ that fits $E$ with $q\subseteq q'$,  there is a  $\preceq^{\mu}$-generalization for $(q,E)$. 
  \item
If there is any CQ $q'$ that fits $E$ with $q'\subseteq q$,  there is a $\preceq^{\mu}$-specialization for $(q,E)$. \end{enumerate}
\end{prop}

On the other hand, there can be infinitely 
many $\preceq^\mu$-repairs. Indeed, it is not difficult to show the following using
a pigeon-hole argument:

\begin{prop}
Let $\calS$ be a schema that contains at least one relation symbol of
arity at least two, and let $\mu$ be any example distribution over $\calS$-examples with
finite support. Then there are annotated CQs $(q,E)$ over $\calS$ for which there are
infinitely many $\preceq^\mu$-repairs, up to equivalence.
\end{prop}

\begin{proof}
    Let $e_1, \ldots, e_n$ be the examples in the support of $\mu$.
    Let $R\in\calS$ be a relation symbol of arity at least two. There are
    infinitely many pairwise non-equivalent CQs over $\calS$, for instance the CQs whose
    canonical instances contain a directed $R$-path of length $m$ in the first two
    coordinates of $R$, with pairwise distinct fresh variables in all remaining
    coordinates (and with an additional $R$-atom for each answer variable, in order to
    satisfy the safety condition). Each CQ over $\calS$ induces a labeling of
    $e_1, \ldots, e_n$ as positive and negative examples, and there are only $2^n$ such
    labelings. Therefore, by the pigeon hole principle, there is a labeling $E$ of
    $e_1, \ldots, e_n$ that is induced by infinitely many pairwise non-equivalent CQs;
    each of these CQs fits $E$. Let $q$ be the maximally-constrained CQ $q_\bot$ over
    $\calS$ of the relevant arity (cf.~Theorem~\ref{thm:dom-generalization-most-specific}),
    so that every relation symbol of $\calS$ occurs in $q$, in accordance with
    Remark~\ref{rem:no-other-relations}. All CQs that fit $E$ agree with each other on
    every example in the support of $\mu$, and hence they all have the same distance to
    $q$. Consequently, each of them is a $\preceq^{\mu}$-repair for $(q,E)$, and there are
    infinitely many of them up to equivalence.
\end{proof}

In light of this, we omit a complexity-theoretic analysis for 
$\preceq^\mu$-repairs.

\section{Discussion}
\label{sec:discussion}

We proposed and studied notions
of \emph{$\preceq$-generalizations}, \emph{$\preceq$-specializations}, and \emph{$\preceq$-repairs}, parameterized by a proximity pre-order $\preceq$, providing  a  
principled framework
for example-driven query debugging and refinement (an idea partially inspired by the interactive schema-mapping design tool  EIRINI~\cite{Alexe2011:designing}).
    We explored two ways to obtain a proximity pre-order for CQs: 
containment-based and edit distance-based. In each case, we assessed the behavior of the obtained repair notions through examples, and we studied
the existence, verification and construction problems, as well as the size of repairs.
Other algorithmic problems may be considered in  follow-up work, such as 
repair enumeration (cf.~\cite{Kimelfeld2020:counting}) and  computing ``possible'' or ``certain'' answers across all repairs of a given CQ (as in~\cite{Barcelo2020:static} for query approximations). 

Based on our findings, $\preceq^{\dom}$-generalizations
and $\preceq^{\dom}$-specializations are reasonably well-behaved (although
the latter do not always exist) while unconstrained $\preceq^{\dom}$-repairs are not;
$\preceq^{\editdist}$-repairs behave   
favorably: they exist whenever a fitting CQ exists; there are always only finitely many repairs up to equivalence; and the main algorithmic tasks are decidable with typically lower complexity than for $\preceq^{\dom}$-repairs; in addition, $\preceq^{\editdist}$-repairs tend to be of smaller size (Proposition~\ref{prop:size-of-editdist-repairs}). 
Although $\preceq^{\editdist}$-repairs, too, in some cases
 display surprising behavior (Example~\ref{ex:editdist-positive-not-a-generalization}), this may be an unavoidable consequence of our syntax-independence requirement  taken together with the inherently syntactic nature of edit distance.
 
An important outstanding issue with 
$\preceq^{\editdist}$-repairs is to design practical algorithms for constructing them. In \cite{ijcai2023:satbased}, a SAT-based
 approach for computing minimal-size fitting $\mathcal{ELI}$-concepts (i.e., Berge-acyclic connected unary CQs) was proposed and evaluated, showing promising performance. While a SAT-solver 
 may not be applicable here due to
 the higher complexity of the problem, we
 believe that inspiration can be taken from this approach. It may also be worthwhile to study approximation algorithms for distance-based query repairs,
i.e., algorithms that produce fitting CQs that have near-minimal distance to the input CQ.

Naturally, these results are specific to the particular class of  queries  we considered: CQs.
One may also study query repairs for other query classes (e.g., self-join-free CQs, unions of CQs, or nested queries). Note, in particular, that self-join-free CQs are cores and our examples of counterintuitive behavior of $\preceq^{\editdist}$-repairs  hinge on the fact that adding atoms to, or removing atoms from, a core CQ need not yield a core CQ (Examples~\ref{ex:editdist-positive-not-a-generalization} and~\ref{ex:editdist-negative-not-a-specialization}). Moreover,
the computational problems associated with $\preceq^{\editdist}$-repairs are
often of lower complexity for CQs that are cores.
For self-join-free CQs, it may moreover be natural to consider, besides
the addition and the deletion of atoms, edit operations that modify variable occurrences,
for instance the replacement of an occurrence of a variable $u$ by a variable $x$. This
leads to refinements of edit distance in the spirit of
Remark~\ref{rem:finer-edit-distance}, whose study we leave for future work.


It is also natural to let the space of candidate repairs depend on the input query in a stronger way. For instance, we may require that
the join structure of the query remains fixed, and  that
repairs differ only in their WHERE-clause (cf.~\cite{DBLP:conf/aaai/MusleaL05}).

Among different avenues for further research, let us mention one: query repair operations could be studied from a more structural perspective. For instance,
how does repairing a query with respect to a collection of labeled examples $E\cup E'$ compare to repairing it with respect to $E$ followed by repairing it with respect to $E'$? And, assuming a ``true'' CQ $q^*$, if we repair a given CQ $q$ repeatedly based on labeled examples for $q^*$, does this process converge towards $q^*$ in a formal sense?

\bibliographystyle{alphaurl}
\bibliography{bib}

\end{document}